\documentclass{article}
\usepackage{booktabs} 
\usepackage{url}

\usepackage[ruled]{algorithm}
\usepackage[noend]{algpseudocode}
\usepackage{amsmath}
\usepackage{amssymb}
\usepackage{amsthm}
\usepackage{color}
\usepackage{graphicx}
\usepackage{epsfig}
\usepackage{epstopdf}
\usepackage{amsfonts}
\usepackage{latexsym}
\usepackage{amssymb}
\usepackage{amsmath}
\usepackage{multirow}
\usepackage{mathtools}
\usepackage[export]{adjustbox}

\usepackage{framed}

\usepackage{comment}
\usepackage{fullpage}
\usepackage[table]{xcolor}

\usepackage{hyperref}
\usepackage{xspace}
\usepackage{authblk}

\renewcommand{\algorithmicrequire}{\textbf{Input:}}

\newcommand{\Input}{\renewcommand{\algorithmicrequire}{\textbf{Input:}} \Require}

\newcommand{\Init}{\renewcommand{\algorithmicrequire}{\textbf{Initialize:}} \Require}
\newcommand{\Begin}{\renewcommand{\algorithmicrequire}{\textbf{Begin:}} \Require}

\newtheorem{claim}{Claim}

\newtheorem{definition}{Definition}

\newtheorem{lemma}{Lemma}
\newtheorem{theorem}{Theorem}

\newcommand{\Alg}{\mathcal{A}}
\newcommand{\AllMakeIt}{\textsc{AllMakeIt}}
\newcommand{\AllButOne}{\textsc{AllButOne}}
\newcommand{\traveller}{agent }
\newcommand{\travellers}{agents }
\newcommand{\BS}{Bike Sharing\xspace}
\newcommand{\RBS}{Relaxed Bike Sharing\xspace}
\newcommand{\BSP}{\textsc{BS}\xspace}
\newcommand{\RBSP}{\textsc{RBS}\xspace}
\newcommand{\algoprob}[4]{\vspace{3mm}{\bf \large {#1}} \\ \noindent{\bf Input: } {#2} \\ \noindent{\bf Output: } {#3} \\ \noindent{\bf Objective: } {#4}\vspace{3mm}}

\DeclareMathOperator{\poly}{poly}

\title{The \BS Problem}

\author[1*]{Jurek Czyzowicz}
\author[2*]{Konstantinos Georgiou}
\author[3*]{Ryan Killick}
\author[3*]{Evangelos Kranakis}
\author[4]{Danny Krizanc}
\author[5*]{Lata Narayanan}
\author[5*]{Jaroslav Opatrny}
\author[5*]{Denis Pankratov}

\affil[1]{D\'{e}partemant d'informatique, Universit\'{e} du Qu\'{e}bec en Outaouais,  Gatineau, Canada}
\affil[2]{Department of Mathematics, Ryerson University, Toronto, Canada}
\affil[3]{School of Comp. Sci., Carleton University, Ottawa, Canada}
\affil[4]{Department of Mathematics \& Comp. Sci., Wesleyan University, Middletown CT, USA}
\affil[5]{Department of Comp. Sci. and Software Eng., Concordia University, Montreal, Canada}
\affil[*]{Research supported in part by NSERC}

\date{}

\begin{document}
\maketitle
\begin{abstract}
    Assume that $m \geq 1$  autonomous mobile agents and $0 \leq b \leq m$ single-agent transportation devices (called {\em bikes})  are initially placed at the left endpoint $0$ of the unit interval $[0,1]$. The agents are identical in capability and can move at speed one. The bikes cannot move on their own, but any agent riding bike $i$ can move at speed $v_i > 1$. An agent may ride at most one bike at a time. The agents can cooperate by sharing the bikes; an agent can ride a bike for a time, then drop it to be used by another agent, and possibly switch to a different bike. 
        
    We study two problems. In the \BS problem, we require all agents and bikes  starting at the left endpoint of the interval to reach the end of the interval as soon as possible. In the \RBS problem, we aim to minimize the arrival time of the agents; the bikes can be used to increase the average speed of the agents, but are not required to reach the end of the interval.
        
    Our main result is the construction of a polynomial time algorithm for the \BS problem that creates an arrival-time optimal schedule for travellers and bikes to travel across the interval. For the \RBS problem, we give an algorithm that gives an optimal solution for the case when at most one of the bikes can be abandoned.
\end{abstract}


\section{Introduction}\label{sec:intro}
    Autonomous mobile robots are increasingly used in many manufacturing, warehousing, and logistics applications. A recent development is the increased interest in deployment of so-called {\em cobots} - collaborative robots - that  are intended to work collaboratively and in close proximity with humans \cite{wikiCobot2020,Boy2007,veloso2012}. Such cobots are controlled by humans, but are intended to augment and enhance the capabilities of the humans with whom they work. 
  
    In this paper we study applications of cobots to transportation problems.  We propose  a new paradigm of cooperative transportation  in which cooperating autonomous mobile agents (humans or robots) are aided by  transportation cobots  (called bikes) that increase the speed of the \travellers when they are used. The \travellers are autonomous, identical in capability, and  can {\em walk} at maximum speed $1$. The bikes are not autonomous and cannot move by themselves, but \travellers can move bikes by riding them. At any time, an agent can ride at most one bike, and a bike can be ridden by at most one agent. An \traveller  riding bike $i$ can move at speed $v_i > 1$; note that the speeds of the bikes may be different. The bikes play a dual role: on the one hand, they are {\em resources} that can be exploited by the \travellers to increase their own speed, but on the other hand, they are also {\em goods} that need to be transported. We assume that initially all \travellers and bikes are located at the start of the unit interval $[0, 1]$. The goal of the \travellers is to traverse the interval, while also transporting (and being transported by) the bikes, in such a way as to minimize the latest arrival time of \travellers and bikes at the endpoint 1. 
        
    If the number of the bikes is smaller than the number of agents, or if the speeds of the bikes are different, it is clear that \travellers can collectively benefit by  {\em sharing}  the bikes. Consider for example, the simple case of 2 \travellers and 1 bike with speed $v>1$. If one of the \travellers uses the bike, and the other walks for the entire interval, the latest arrival time of the \travellers is 1. However if one of the \travellers rides the bike for distance $x$, drops the bike, and then walks the remaining distance, while the second \traveller walks until it reaches the dropped-off bike, and rides it for the remaining distance, the latest arrival time is $\max\{x/v + 1-x, x + (1-x)/v \} $ which is clearly less than 1. It is easy to see that the optimal value of $x$ is $1/2$ giving an arrival time of $\frac{1+v}{2v}$ for the two \travellers and a bike. 

    From the above example, we can see that a solution to the cooperative transport problem described here consists of a {\em schedule:} a partition of the unit interval into sub-intervals,   as well as a specification of the bike to be used by each \traveller in each  sub-interval. It is easy to see that it is never helpful for agents to turn back and ride in the opposite direction. In Section~\ref{sec:no-wait} we show that it is also not helpful for agents to {\em wait} to switch bikes, or to ride/walk at less than the maximum speed possible.  Thus we only consider {\em feasible} schedules wherein  the bike to be used by an agent in a sub-interval is always available when the agent reaches the  sub-interval. This will be defined formally in Section~\ref{sec:prelim}; for now it suffices to define the  {\em \BS} ($\BSP$ for short) problem informally: Given $m$ agents and  $b$ bikes, and the speeds of bikes $v_1 \ge v_2 \ge \cdots \ge v_b > 1$, find a feasible schedule specifying how $m$ agents can travel and transport bikes from initial location $0$ to location $1$ by sharing bikes, so that the  latest arrival time of agents and bikes is minimized.

    We also consider a variation of the problem when there is no requirement for the bikes to be transported; bikes are simply resources available to the \travellers to use in order to minimize the agents' arrival time, and can be abandoned  enroute if they are not useful towards achieving this objective. Thus the {\em \RBS} ($\RBSP$ for short) problem is: Given $m$ agents and  $b$ bikes, with respective speeds $v_1 \ge v_2 \ge \cdots \ge v_b > 1$, find a feasible schedule specifying how $m$ agents can travel  from initial location $0$ to location $1$ by sharing bikes,  and the  latest arrival time of agents is minimized. We also consider a version of the problem when the input also specifies an upper bound $\ell$ on the number of bikes that can be abandoned.

    Notice that for the same input, the optimal arrival time for the \RBSP problem can be less than that for the \BSP problem, as it may sometimes be more advantageous to abandon the slower bikes at some point and share the faster bikes. For example, suppose we have two agents and 2 bikes with speeds $v_1, v_2$ with $v_1 > v_2 > 1$. For the \BSP problem, the best arrival time achievable is $1/v_2$, as both bikes have to make it to the end. But for the \RBSP problem, a better arrival time for the agents can be achieved by the following strategy (see Figure~\ref{fig:22}): one agent rides the faster bike up to some point $z$ and then walks the remaining distance taking time $z/v_2 + 1 - z$. The other agent rides the slower bike up to point $z$, abandons it, switches to the faster bike, and rides it to the end. This takes time $z/v_2 + (1-z)/v_1$. By equating these two times, we get an overall arrival time of $\frac{v_1^2 - v_2}{v_2v_1^2 + v_1^2 - 2v_1v_2} \leq \frac{1}{v_2}$.

    \begin{figure}[H]
            \centerline{\includegraphics[scale=0.06]{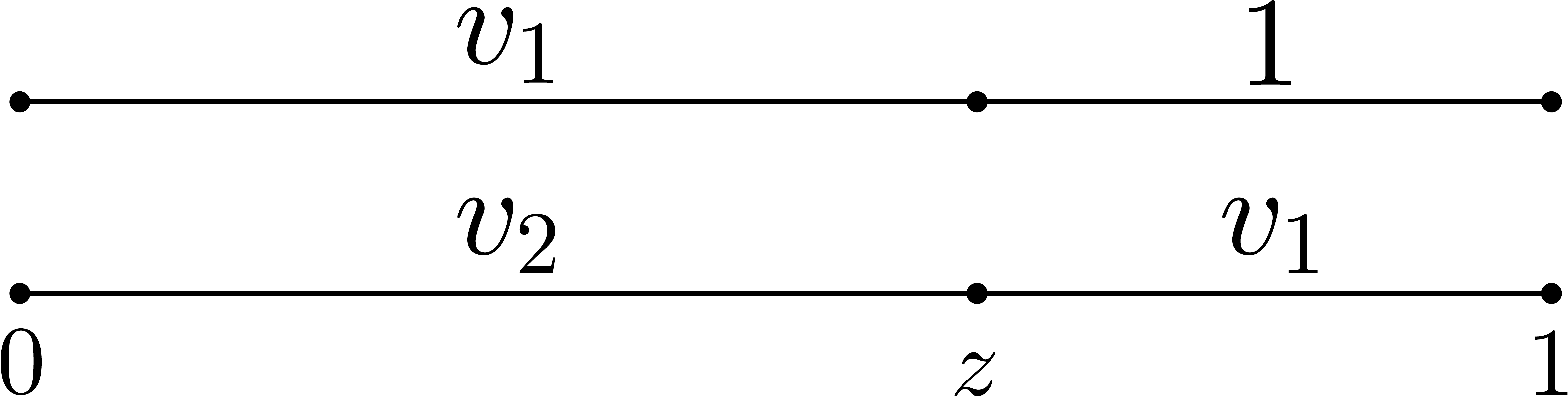}}
            \caption{Strategy for \RBSP for $2$ agents and $2$ bikes. The two rows correspond to the behaviour of two agents. Each line segment is labelled with a bike speed that corresponds to the agent travelling on that bike for the distance of the line segment. If an agent  walks then the line segment is labelled with $1$. Notice that bike 2 is abandoned before reaching the end.  }
            \label{fig:22}
    \end{figure}

\subsection{Related work}\label{ssec:related_work}
    The capabilities of autonomous mobile robots have been extensively studied; many different models have been proposed, algorithms and impossibility results have been given for tasks such as gathering, pattern formation etc. under different assumptions (for related references see the survey papers on pattern formation by \cite{PrencipeMAC19,VigliettaMAC19} and related work cited therein, on swarm robotics by \cite{KingMAC19}, as well as the monograph ~\cite{flocchini2019distributed}).  These robots are generally assumed to be identical, anonymous, and work in discrete look-compute-move cycles. 
        
    Closer to our work is the model used for multi-robot search and evacuation. Indeed scheduling and cooperation is paramount for the optimal performance of an algorithm for {\em evacuation} in a continuous domain, whereby all the robots are required to evacuate from an exit placed at an unknown location on the boundary, and as in our problem, we are interested in minimizing the time the last robot arrives at the exit ~\cite{CGKMAC19}. However, in this problem, unlike our problem,  the destination/exit is at an unknown location. 

    Another related problem is the {\em beachcombers} problem which was first introduced and analyzed in \cite{czyzowicz2015beachcombers}. In this problem,  $n$ mobile robots placed at an endpoint of the unit interval $[0,1]$. Each robot has its own search and walking capability independently of the rest of the robots, namely a searching speed $s_i$ and a walking speed $w_i$, where $s_i < w_i$, for $i=1,2,\ldots ,n$. The robots are to explore the interval collectively, and the goal is to minimize the time when every point in the interval has been searched by at least one robot.  However, in this  problem, unlike in our problem,   all robots are not required to reach the endpoint of the interval in the beachcombers problem, and also, robots can switch speeds independently of each other, and whenever they want.

    The bike sharing problem involves a sort of cooperation between \travellers and bikes in order to optimize the arrival time. In this context, related studies concern cooperation between mobile agents and wireless sensors. For example, \cite{DSKZ06} and \cite{SSL07} use information obtained from wireless sensors for the problem of localization of a mobile robot while in \cite{JS01}, mobile robots and stationary sensors cooperate in a target tracking problem whereby stationary sensors track moving sensors in their sensor range and mobile agents explore regions not covered by the fixed sensors.  There are also evaluation platforms for mixed-mode environments incorporating both mobile robots and wireless sensor networks, see e.g.,~\cite{Kropff08}. 
    
    Additional literature on cooperation between mobile agents and wireless sensors includes \cite{DBLP:conf/waoa/CzyzowiczKKNO14} in which a mobile robot is assisting moving sensors to attain coverage and \cite{wid2013} which provides hardness results, exact, approximation, and resource-augmented algorithms for determining whether there is a schedule of agents' movements that collaboratively deliver data from specified sources of a network to a central repository.  Of related interest is the addition of immobile~\cite{kranakis2003mobile} and mobile~\cite{czyzowicz2008power} tokens to aid in the exploration considered in the context of the rendezvous problem on a discrete ring network, however in both instances the token is merely used as a marker to recall the presence of another agent in the most recent past.  


    Bike sharing systems have been installed in many cities as a low-cost and environmentally friendly transportation alternative; researchers have considered optimization problems such as placement of bike stations or rebalancing of inventory - see for example \cite{chen15,schuijbroek17}. But this line of work is not concerned with the optimal usage of bikes by a particular set of users to travel between two points. The authors of \cite{Li19} consider {\em trip planning} for a set of users in a bike sharing system, but the problem is quite different from ours: bikes are available at a predefined set of locations, each user has a different start and end location, and can use at most one bike in their trip, and the problem is concerned with optimizing the number of served users and minimizing their trip time. The authors show NP-hardness of their trip planning problem, and give approximation algorithms. Another line of research relates to optimization of ride-sharing/carpooling which aims to match drivers offering rides with people needing rides; see \cite{agatz12} for a recent survey. Once again, each rider uses a single vehicle for their trip; there is no notion of switching vehicles during one trip. A couple of papers did consider ridesharing with transfers, but the algorithms given are heuristic algorithms with no guarantees of quality \cite{coltin2013,drews2013}.

    Our problem has similarities to job shop  scheduling problems \cite{Garey76}, where several jobs of the same processing time are to be scheduled on a collection of machines/servers with different processing power, and the objective is to minimize the {\em makespan}, the time when all jobs have finished.  Indeed our agents can be seen as akin to jobs, we can assume that there is an unlimited number of basic servers of low processing power available (this is akin to agents always being able to make progress by walking), and bikes are smiliar to the fast servers. However, the constraints imposed by bikes being available only at the points they are dropped off do not seem to translate to any natural precedence constraints in the context of multiprocessor or job shop scheduling. To the best of our knowledge, a  scheduling analog of our problem has not been studied.
    
%

\subsection{Outline and results of the paper}\label{ssec:outline}  

    In this paper we introduce and study  variations of the  Bike Sharing problem -- a novel paradigm of mobile agent optimization problems that includes both active elements (agents) and passive elements (bikes). For ease of exposition, we assume that all agents always move at the maximum speed possible (walking or biking), and never wait. In Section~\ref{sec:no-wait} we show that this assumption does not lose any generality. 
    For the \BSP problem, all bikes are required to make it to the end and so the speed of the slowest bike imposes a lower bound on the arrival time. We prove another lower bound on the arrival time based on the ``average'' completion time of the agents. We design an algorithm that produces a schedule with completion time matching the maximum of the two lower bounds. Our main result is the following:

    \begin{theorem}
            There is a polynomial time algorithm that constructs an arrival-time optimal schedule for the \BSP problem.         
    \end{theorem}
    \begin{proof}
            Follows from Theorems~\ref{thm:bs_ub_ge_tmu} and \ref{thm:bs_ub_le_tmu}.
    \end{proof}
    For the variation of the \RBSP problem, in which $\ell$ bikes are allowed to be abandoned, we use similar techniques to show the following:
    \begin{theorem}
            There is a polynomial time algorithm that constructs an arrival-time optimal schedule for the \RBSP problem when at most one bike can be abandoned.     
    \end{theorem}
    \begin{proof}
            Follows from Theorems~\ref{thm:rbs_ub_lt_tmu}, \ref{thm:rbs_ub_gt_tmu_ub1_lt_t1mu} and \ref{thm:rbs_ub_gt_tmu_ub1_gt_t1mu}. 
    \end{proof}
    Finally, we show that the \RBSP problem can be solved efficiently and optimally for some special cases of bike speeds; we state the result  informally here; see Section~\ref{sec:rbs} for details.         
    \begin{theorem}[informal]\label{thm:informal}
            There is a polynomial time algorithm that constructs an arrival-time optimal schedule for the \RBSP problem when at most one bike is ``slow."    
    \end{theorem}

    The outline of the paper is as follows. In Section~\ref{sec:prelim} we introduce an intuitive matrix representation of bike sharing schedules and formally define the problems we are studying.  In Sections~\ref{sec:bs} and ~\ref{sec:rbs} we give our algorithms for the \BSP and \RBSP problems respectively. We then justify our initial assumption that agents always move at the maximum speed and never wait in Section~\ref{sec:no-wait}. Finally, we conclude the paper in Section~\ref{sec:conclusion} with a discussion of open problems.

\section{Definitions and preliminary observations}\label{sec:prelim}
 
    We assume there are $m$ agents and $b$ bikes with speeds $> 1$ on a line. The bikes and agents are initially located at the origin $0$.  All agents have the same walking speed of $1$ and will travel at speed $v_i>1$ when using the bike $i$. An agent can ride at most one bike at a time, and a bike can be ridden by at most one agent at a time. An agent can use a  bike for a portion of the trip and at any time  drop the bike. If an agent comes across bike $i$ while walking or riding a different bike $j$, it may choose to pick up bike $i$ (and drop bike $j$), and continue the trip on bike $i$. We assume that  pickups and dropoffs happen instantaneously, and that   any number of agents and bikes can be present at the same point simultaneously. Thus, bikes can be dropped without blocking other bikes or agents, and agents and bikes can pass each other at will. In the sequel we will assume that the agents always move at the maximum speed allowed by walking/biking, and never stop and wait. 

    When necessary, we will use $v_i$ to represent the speed of the $i^{th}$ bike, however, we will find it much more useful to work with the inverse speeds of the bikes -- which we represent by $u_i \equiv 1/v_i$. Thus, we will assume that the (multi)set $U$ lists the inverse speeds of the bikes, i.e. $U = \{u_1,u_2,\ldots,u_b\}$. An input to both \BSP and \RBSP problems is then represented by a pair $(m,U)$ where $m$ is the number of agents and the number of bikes is $b=|U|$. 
    We will assume without loss of generality that the bikes are labelled in increasing order of their inverse speeds (decreasing order of speeds). Thus, bike $1$ is the fastest bike and bike $b$ is the slowest bike. We reserve label $0$ to denote absence of a bike, i.e., walking. Thus, the inverse speed associated with this label is $1$.

    We say that bike $i$ is \emph{dropped} at location $y \in [0,1]$ if an agent leaves the bike at location $y$ and it is picked up by another agent at a later time. We say that bike $i$ is \emph{abandoned} at location $y$ if an agent leaves the bike at location $y$ and it is never picked up by another agent. Note that bikes cannot be abandoned in the \BSP problem but may  be abandoned in solutions to the \RBSP problem. 

    A solution to either problem is a schedule  that should specify for each agent $i$ a partition of the entire interval $[0,1]$ and for each block of the partition whether the agent walks or uses a particular bike. By taking a common refinement of all these partitions we may assume that all partitions are exactly the same. The information about which bike gets used during which block of the partition and by which agent is collected in a single matrix. Thus, we define a \emph{schedule} as a pair $(X,M)$ where
    \begin{itemize}
        \item $X = (x_1, x_2, \ldots, x_n)$ is a \emph{partition vector} satisfying $\forall j \in [n] \;\; x_j \ge 0$ and $\sum_{j = 1}^n x_j = 1$;
        \item $M$ is a \emph{schedule matrix} of size $m \times n$ such that entry $M(i, j)$ indicates the label of the bike used by agent $i$ during the interval $x_j$ or $0$ if agent $i$ walked during the interval $x_j$.
    \end{itemize}
    By a slight abuse of notation, we use $x_j$ to refer to both the $j^\text{th}$ interval of the partition as well as its length. 

    We define the \emph{size of a schedule} as $n$ -- the number of columns in the schedule matrix $M$ (also the number of entries in the partition vector). 
    While the number of rows $m$ will be the same for any schedule, a priori it is not even clear that there has to be an optimal schedule with a finite number of columns. Clearly, it is desirable to minimize the size of a schedule. At the very least, the size should be at most polynomial in the number of agents, i.e. $n = O(\poly(m))$.

    The matrix $M$ gives rise to an induced matrix of inverse speeds $\widetilde{M}$. More specifically, for $i \in [m]$ and $j \in [n]$ the entry $M(i,j) = k$ implies $\widetilde{M}(i,j) = u_k$ when $k \in [b]$, and $M(i,j) = 0$ implies $\widetilde{M}(i,j) = 1$. 
    It will be convenient for us to treat these representations as interchangeable, which is easy to implement. 

    The main utility of the above definitions is that they allow one to easily compute the completion times of the agents. Indeed, the completion time $t_i$ of agent $i$ is computed as the $i^\text{th}$ row of the matrix product $\widetilde{M} \cdot X$, i.e.
    \begin{equation}\label{eq:ti_general}
        t_i(X,M) = \sum_{j=1}^{n} \widetilde{M}(i,j) x_j.
    \end{equation}
    The time $t_{i,j}$ at which agent $i$ reaches the end of interval $x_j$ can be similarly computed using
    \begin{equation}\label{eq:tij_general}
        t_{i,j}(X,M) = \sum_{k=1}^{j} \widetilde{M}(i,k) x_{k}.
    \end{equation}
    We call $t_{i,j}$ the $j^{th}$ {\em completion time} of agent $i$. Note that $t_i = t_{i,n}$. The overall completion time will be represented by $\tau$ and equals the maximum of the $t_i$,
    \begin{equation}\label{eq:tau_general}
        \tau(X,M) \equiv \max_i t_i.
    \end{equation}
    We may omit $(X,M)$ when it is clear from the context and simply write $t_i, t_{i,j},$ and $\tau$. We shall sometimes refer to the completion time as \emph{arrival time}.

    Not every schedule $(X,M)$ is feasible. In particular, feasibility requires that (1) a bike cannot move on its own, (2) a bike is used by at most one agent during a particular interval of the partition, and (3)  a bike is available when an agent is supposed to use that bike. In light of the introduced notation and terminology we can formalize the above three conditions as follows:
    \begin{definition}\label{def:feasible}
        $(X,M)$ represents a feasible schedule if:
        \begin{enumerate}
            \item $M(i,j) \neq 0$ implies that $\exists\ i'\;\; {M}(i',j-1) = {M}(i,j)$,
            \item ${M}(i,j) \neq 0$ implies that $\forall\ i' \neq i \;\; {M}(i',j) \neq {M}(i,j)$,
            \item ${M}(i,j) = {M}(i',j-1) \neq 0$ implies that $t_{i',j-1} \leq t_{i,j-1}$.
        \end{enumerate}        
    \end{definition}        

    In a feasible schedule $(X,M)$ we say that \emph{bike $k$ makes it to the end} if this bike is used during the last interval, i.e., $\exists i$ such that ${M}(i, n) = k$.
        
    We are now ready to state the \BSP and \RBSP problems formally:

    \algoprob{The \BS Problem}{$m \ge 1$ -- number of agents; $b \ge 0$ -- number of bikes; $0 < u_1 \le u_2 \le \cdots \le u_b < 1$ -- inverse speeds of bikes.}{A feasible schedule $(X,M)$ such that every agent and bike makes it to the end.}{Minimize $\tau(X,M)$.}

    \algoprob{The \RBS Problem}{$m \geq 1$ -- number of agents; $b \geq 0$ -- number of bikes; $0 < u_1 \le u_2 \le \cdots \le u_b < 1$ -- inverse speeds of bikes.} {A feasible schedule $(X,M)$ such that all agents reach the end.}{Minimize $\tau(X,M)$.}

    Since many different schedules may achieve the same completion time, it is important to consider some complexity measures other than computational time that can be used to compare schedules. We already defined the {\em size} of a schedule and it can be considered as one such measure. Another measure is the {\em switch complexity} of a schedule, which is defined as the total number of {\em bike switches} in the schedule. We distinguish between two different kinds of switches: (1) \emph{a drop-switch} when an agent drops a bike and starts walking, and (2) \emph{a pickup-switch} when an agent picks up a new bike and the agent was either walking or using a different bike immediately prior to the switch. 

    Formally, we say that agent $i$ performs \emph{a drop-switch} between the intervals $x_{j-1}$ and $x_j$ if ${M}(i,j)=0$ and ${M}(i,j-1)\neq 0$.  We say that agent $i$ performs a \emph{pickup-switch} switch with agent $i' \neq i$ between the intervals $x_{j-1}$ and $x_j$ if ${M}(i,j) = {M}(i',j-1) \neq 0$.

    So far we have stated that the agents have to traverse the interval $[0,1]$. We could instead consider schedules for intervals of the form $[0,a]$ with obvious modifications to the definitions of a schedule and feasibility. Due to the linear nature of constraints and arrival times, a feasible schedule for the interval $[0,a]$ can be easily transformed into a feasible schedule for another interval $[0,b]$ by scaling:

    \begin{claim}\label{claim:scaling}
        Let $a,b \in \mathbb{R}_{>0}$ and let $(X,M)$ be a feasible schedule for the interval $[0,a]$. Define $X'$ to be the scaled version of $X$:
        $X' = (b/a) X$.
        Then $(X', M)$ is a feasible schedule for the interval $[0,b]$. Moreover, for all $i$ we have $t_i(X',M) = (b/a) t_i(X, M)$. In particular $\tau(X',M) = (b/a) \tau(X,M)$.
    \end{claim}

    We can shift starting locations of agents and bikes (this obviously doesn't affect the arrival times) and work with schedules for arbitrary intervals $[a,b]$ such that $a < b$. First of all, these observations show that our focus on interval $[0,1]$ is without loss of generality. Second of all, our algorithms make an extensive use of these observations when recursively created schedules get used in certain subintervals of $[0,1]$.

    For the presentation of our algorithms in pseudocode, we introduce the following notation. $M(1:m,j)$ is used to indicate the $j^\text{th}$ column of $M$ and $M(i,j:n)$ is used to indicate the suffix of length $n-j+1$ of row $i$ of $M$. The notation $[X'; X(j)]$ is used to indicate the (column) vector $X'$ with the entry $X(j)$ adjoined to its end. Similarly, the notation $[M', M(1:m,j)]$ is used to indicate the matrix $M'$ with the column $M(1:m,j)$ adjoined to its end.

    We finish this section with an important observation. Notice that a solution to the \BSP problem naturally has two components: the continuous component represented by the partition vector, and the discrete component represented by the schedule matrix. There is a tension between these components. On one hand, the continuous component suggests that we may try approaching the \BSP problem using powerful tools  such as linear or convex programming. On another hand, the number of possible schedule matrices is exponentially large and this suggests that maybe hardness of some combinatorial optimization problems might be encoded in schedule matrices. Moreover, these components are not independent and they have interesting interactions. These observations suggest that the core of the problem lies in finding an optimal discrete component and then figuring out its interaction with the continuous component. Thus, intuitively we should focus on designing a ``good'' schedule matrix $M$. We formalize this observation as follows: 

    \begin{lemma}\label{lm:semi_opt}
        For any schedule matrix $M$  satisfying parts 1 and 2 of Definition~\ref{def:feasible}, we can in polynomial time find a partition vector $X$ that makes schedule $(X_M,M)$ feasible and achieves the smallest completion time among all schedules with the same schedule matrix $M$. 
    \end{lemma}
    \begin{proof}
        Fix $M$ and compute an optimal choice of $X$ for the given $M$ by solving the following optimization problem:
        \begin{align}\label{eq:LP2}
            \mbox{Minimize: }& \tau \nonumber\\
            \mbox{Subject to: }& \tau \geq t_i,& 1 \leq i \leq m \nonumber\\
            &t_{i',j-1} \leq t_{i,j-1}& \mbox{ whenever } M(i,j)=M(i',j-1)\neq 0,\ i\neq i'\\
            &x_j \geq 0,& 1\leq j \leq n,\nonumber\\
            & \sum_{j=1}^n x_j = 1.\nonumber
        \end{align}
        The constraints encode that $X$ must be a partition of the interval $[0,1]$ and that schedule $(X,M)$ must satisfy Definition~\ref{def:feasible} (notice that other parts of Definition~\ref{def:feasible} refer only to $M$). The objective is to minimize latest completion time. Observe that $t_i$ and $t_{i,j}$ are not variables in the above problem, but are just shorthands for the linear functions of variables $X$ as defined by Equation~\eqref{eq:ti_general} and \eqref{eq:tij_general}.
    \end{proof}
    We note that it is not guaranteed that an optimal solution to the LP \eqref{eq:LP2} is unique. However, at least one of these optimal solutions is guaranteed to reside at a vertex of the feasible region of \eqref{eq:LP2}. Let  $X_M$ be an optimal solution that corresponds to an optimal vertex of the LP \eqref{eq:LP2}. We will call $X_M$ an {\em induced partition vector} of $M$.

\section{Optimal and Efficient Algorithm for the \texorpdfstring{\BSP}{BS} Problem}\label{sec:bs}
    In this section we present a polynomial time algorithm that solves the \BS problem optimally, i.e., our algorithm produces a feasible schedule that minimizes the latest arrival time. 

    \emph{High-level overview of the algorithm.} The structure of our algorithm is quite involved and it relies on several subroutines, so we do not present the entire algorithm all at once. We build up towards the complete algorithm in stages:
    \begin{enumerate}
        \item We begin by deriving a lower bound $T$ on the minimum completion time of any feasible schedule. It turns out the relationship between the inverse speed of the slowest bike $u_b$ and $T$ controls the structure of an optimal schedule matrix for the \BS problem. 
        \item We first tackle the case when $u_b \le T$. We describe a subroutine \textsc{AllMakeIt} that calls itself recursively to build up an optimal schedule matrix for the case of $u_b \le T$. We prove the optimality of this subroutine and compute the size of the final schedule that this subroutine outputs. 
        \item Unfortunately, the \textsc{AllMakeIt} subroutine creates a schedule of size that is exponential in the number of bikes. We describe two more subroutines -- \textsc{Standardize} and \textsc{Reduce} that allow transforming any schedule of potentially large size $n$ into an equally or more efficient schedule of size $n' \le m$.
        \item We finish the case $u_b \le T$ by describing two modifications to \textsc{AllMakeIt} resulting in a polynomial time algorithm. First, we use the \textsc{Reduce} subroutine to prevent the size of the schedule from growing exponentially. Second, we observe that recursive calls have overlapping substructures, so we can rewrite the algorithm with the help of a dynamic programming table. With these modifications, the altered \textsc{AllMakeIt} subroutine runs in polynomial time and produces a schedule of size $\le m$.
        
        \item Then we tackle the case of $u_b > T$ by showing that it can be reduced, in a certain sense, to the case of $u_b \le T$. This completes the description of the algorithm.
    \end{enumerate}

    \subsection{Lower Bound on the Arrival Time for the \BS Problem}
        Define $T(m, U)$ as follows:
        \begin{equation}\label{eq:TmU}
            T(m,U) := 1 - \frac{1}{m} \sum_{j=1}^b (1-u_j)
        \end{equation}
        The following lemma shows that $T(m,U)$ is a lower bound on the optimum for the \BS problem. For ease of reference, we also record another trivial lower bound in the statement of the lemma, observe a sufficient condition for the schedule to satisfy $\tau(X,M) = T(m,U)$, and give two useful properties that result from the definition of $T(m,U)$.

        \begin{lemma} \label{lem:lbs-for-bsprob}
            Let $(X,M)$ be an arbitrary feasible schedule for the \BS problem with inputs $m, b,$ and $U$. Then we have
            \begin{enumerate}
                \item $\tau(X,M) \ge u_b$,
                \item $\tau(X,M) \ge T(m,U)$,
                \item\label{lempart:same-arrival} $\tau(X,M) = T(m,U)$ if and only if all agents have the same arrival time in $(X,M)$.
                \item\label{lempart:min_u1} $u_1 \leq T(m,U)$.
                \item\label{lempart:TmU_order} $T(m,U) \leq T(m-1,U\setminus \{u_k\})$ iff $u_k \leq T(m,U)$ with equality only when $u_k = T(m,U)$.
            \end{enumerate}
        \end{lemma}

        \begin{proof} 
            \begin{enumerate}
                \item This part follows from the fact that $u_b$ time is required for the slowest bike $b$ to travel from $0$ to $1$, which this bike is required to do.
                \item The total distance travelled by all agents in $(X,M)$ is $m$, since each agent travels distance $1$. Since all bikes must make it to the end,  the total distance travelled by each bike is exactly $1$, which implies that the sum of distances walked by all agents is $m-b$. Since walking speed is $1$, the total time agents spend walking is $m-b$, whereas the total time agents spend on bikes is $\sum_j u_j$. Combining these observations, we obtain:
                    $\sum_{i=1}^m t_i(X,M) = (m-b) + \sum_{j=1}^b u_j = m - \sum_{j=1}^b (1-u_j)$.
                Since the maximum is at least the average, we get
                    $\tau(X,M) \ge \frac{1}{m} \sum_{i=1}^m t_i(X,M) = 1 - \frac{1}{m} \sum_{j=1}^b (1-u_j) = T(m,U)$.
                \item Follows from the fact that the maximum $t_i$ is exactly equal to the average if and only if all values are the same.
                \item Follows from the fact that the minimum $t_i \geq u_1$ and the average is at least the minimum.
                \item Observe that $T(m,U)$ is just the average value of $U \cup \{1,\ldots,1\}_{m-b}$. Thus, this part follows from the fact that the removal from a set a member whose value is above/below/equal to the average value results in a new set with an average value smaller than/larger than/equal to the average value of the original set.
            \end{enumerate}
        \end{proof}

        In spite of the simplicity of the above lower bounds, they turn out extremely important: in the following sections we show that if $u_b \le T(m, U)$ then there is a feasible schedule $(X,M)$ with completion time $\tau(X,M) = T(m,U)$, and otherwise there is a feasible schedule $(X,M)$ with completion time $\tau(X,M) = u_b$. Putting it together the \BS problem has an optimal schedule that has completion time $\max(u_b, T(m,U))$.

\subsection{Finding an Optimal Schedule for the Case \texorpdfstring{$u_b \le T(m,U)$}{ub <= T(m,U)}}\label{ssec:all_make_it}

    The goal of this section is to present and analyze the algorithm \textsc{AllMakeIt} that on input $(m, U)$ satisfying $u_b \le T(m,U)$ produces a feasible schedule $(X,M)$ with $\tau(X,M) = T(m,U)$.

    In anticipation of the recursive nature of this algorithm, we introduce the following notation: for $k \in \{0, 1, 2, \ldots, b\}$ we define 
    \begin{equation}
        m_k := m-b+k,\ \mbox{ and  }\ \  U_k := \{u_1, \ldots, u_k\}.
    \end{equation} 
    Note that we have $U_0 = \emptyset$ and $U_b = U$.

    We say that a group of agents is \emph{synchronized} at a particular time if they are found to be at the same location at that time. We say that agent $i$ is a \emph{walker} at time $t$ in a given schedule if the agent does not use a bike at time $t$. It may be helpful to refer to Figure~\ref{fig:all_make_it} demonstrating the schedule matrix produced by \textsc{AllMakeIt} while reading the following description of the algorithm.

    At a high level, \textsc{AllMakeIt} has two phases. The goal of the first phase is to get a group of $m-b$ synchronized walkers ahead of all the bikes. At that time the bikes will be in use by the remaining $b$ agents. Moreover, the tardiness of agents on bikes has a particular order: if we make no changes to the schedule after the first phase and let the walkers walk and biking agents  bike, then the agent on bike $1$ will catch up to the walkers first, followed by the agent on bike $2$, and so on. This goal can be achieved by sharing bikes among $m$ bikers as follows. At first, agents $1, \ldots, b$ use bikes while others walk. After travelling a certain distance, each agent on a bike drops their bike to be picked up by an agent immediately behind them. This way the set of agents that use bikes during consecutive intervals propagate from $1, \ldots, b$ to $2, \ldots, b+1$ to $3, \ldots, b+2$, and so on until all the bikes accumulate among the last $b$ agents numbered $m-b+1, \ldots, m$. The intervals of the partition vector during this phase can be arranged so that walkers $1, 2, \ldots, m-b$ are synchronized at the end of the first phase.
    \begin{figure}[!h]
        \centering
        \includegraphics[scale=0.42]{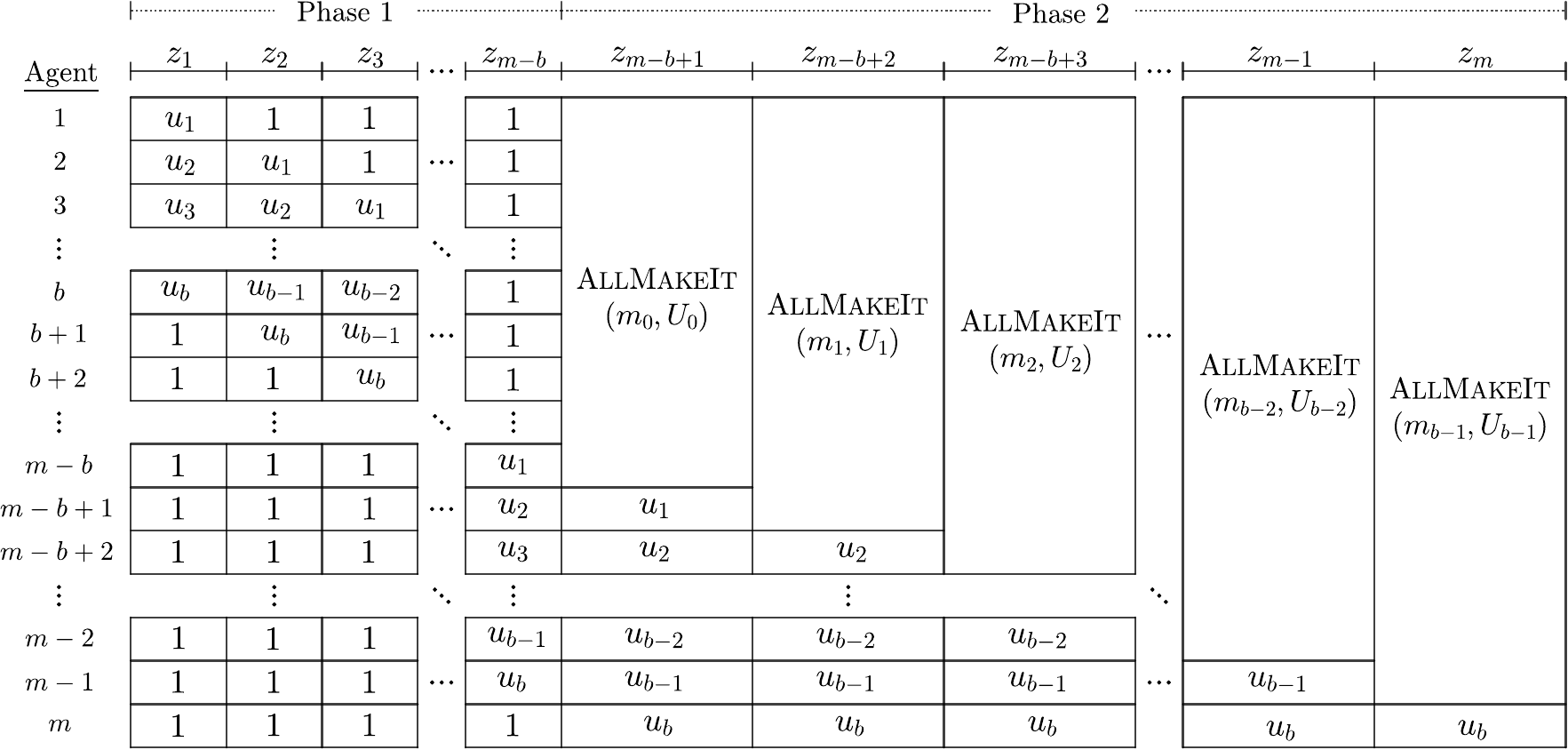}
        \caption{The schedule matrix constructed by $\textsc{AllMakeIt}(m,U)$. Note that $b<m-b$ in this figure. The first phase corresponds to intervals $z_1, \ldots, z_{m-b}$ and the second phase corresponds to the intervals $z_{m-b+1}, \ldots, z_m$. Each large rectangle indicates a recursive call to the algorithm on the indicated number of agents using the indicated set of bikes. Note that the partition $Z$ refers to the intervals before expanding the recursively constructed submatrices. We describe later how to obtain the true partition $X$  after such an expansion.\label{fig:all_make_it}}
    \end{figure}
    
    In the second phase, the agent using bike $1$ is the first agent on a bike to catch up with the group of walkers. Rather than overtaking the group of walkers, the new agent and the bike ``get absorbed'' by the group. Intuitively, it means that now we have a group of $m-b+1$ agents with $1$ bike that can be shared in the entire group to increase the average speed with which this group can travel. It turns out that the best possible average speed of this group is  exactly $1/T(m_1, U_1)$. Moreover, we can design a schedule which says how this $1$ bike is shared with the group during any interval of length $x$ by taking a schedule produced (recursively) by \textsc{AllMakeIt} on input $(m_1, U_1)$ and scaling this schedule by $x$. This has the added benefit that if the group was synchronized at the beginning of the interval (which it was) the group remains synchronized at the end of the interval. The distance $x$ is chosen so the agent using bike $2$ catches up to the new group precisely after travelling distance $x$. When the agent using bike $2$ catches up to the group, the new agent and the new bike again get absorbed by the group, further increasing the average speed. An optimal schedule for the new group is given by a scaled version of  \textsc{AllMakeIt} on input $(m_2, U_2)$. This process continues for bikes $3, 4, \ldots, b-1$. Since recursive calls have to have smaller inputs, we have to stop this process before bike $b$ gets absorbed. What do we do with bike $b$ then? The answer is simple: we arrange the entire schedule in a such a way that bike $b$ catches up with the big group precisely at the the end of the entire interval $[0,1]$. The pseudocode for this algorithm is provided below.

    \begin{algorithm}[!h] \caption{$\textsc{AllMakeIt}(m,U)$} \label{alg:allmakeit}
        \begin{algorithmic}[1]             
        \Input {$m$ -- number of agents, $U = \{ u_1 \le u_2 \le \cdots \le u_b\}$ -- inverse bike speeds; \\
        \textbf{Precondition:} $b=0$ or $u_b \le T(m,U)$}\label{algline:allmakeit-precond}
        \Begin
            \If{$b = 0$}\Comment{Base case: no bikes, schedule is trivial, every agent walks}
            \State $X \gets [1]$ \Comment{One interval in the partition}
            \State $M \gets [0; 0; \ldots; 0]$\Comment{Schedule matrix is $m \times 1$ populated only with $0$ labels}
            \Else
            \State $M$ is initialized to an empty $m \times m$ matrix
            \For{$(j=1$ to $m-b)$} \Comment{Populate part of the matrix corresponding to the first phase}
                \For{$(i=1$ to $m)$}
                \If{$j \le i \le j+b-1$}\Comment{Figure~\ref{fig:all_make_it} may be helpful for understanding this part}
                    \State{$M(i,j) \gets i-j+1$}
                \Else
                    \State{$M(i,j) \gets 0$}
                \EndIf
                \EndFor
            \EndFor
            \For{$(j=m-b+1$ to $m)$} \Comment{Populate part of the matrix corresponding to the second phase}
            \State{$M(1 : (j-1), j) \gets \textsc{AllMakeIt}(j-1, \{u_1, \ldots, u_{j-m+b-1}\})$} \Comment{Recursive call}
                \For{$(i=j$ to $m)$}
                    \State{$M(i,j) \gets i-m+b$}
                \EndFor
            \EndFor
            \State{$M \gets \textsc{Expand}(M)$}\label{algline:allmakeit} \Comment{Expand recursively computed submatrices}
            \State{$X \gets \textsc{SolveLP}(U, M)$} \Comment{Solve the LP \eqref{eq:LP2}}
            \EndIf
            \State \Return $(X,M)$
        \end{algorithmic}
    \end{algorithm}  

    The pseudocode in Algorithm~\ref{alg:allmakeit} essentially populates the schedule matrix $M$ according to Figure~\ref{fig:all_make_it} and then finds an optimal $X$ by solving the LP \eqref{eq:LP2}. We chose this approach for simplicity and because it suffices for our purposes, but one can instead use a more direct and efficient way of computing optimal $X$ for the schedule matrix $M$ which is implicit in the proof of Theorem~\ref{thm:all_make_it} below. In the pseudocode we assign matrices obtained via recursive calls to \textsc{AllMakeIt} to certain ranges inside $M$. Thus, immediately prior to line \ref{algline:allmakeit}, $M$ has this doubly layered structure and $M$ needs to be brought to a proper $2$-dimensional array representation. Suppose that $M(1:i, j)$ contains an $i \times n'$ matrix $M'$. We modify $M$ by inserting $n'$ columns starting at column $j$, populating the first $i$ rows of the newly inserted columns with entries from $M'$ and populating the remaining $m-i$ rows of newly inserted columns by simply replicating  $n'$ times the submatrix $M((i+1):m, j)$. This is repeated for all recursively computed matrices. We refer to this procedure as \textsc{Expand} in line~\ref{algline:allmakeit} in the pseudocode. This procedure should be straightforward to implement, but details are a bit messy so we omit them for ease of presentation.    
    
    Before we dive into the analysis of the algorithm, we verify that the precondition in line~\ref{algline:allmakeit-precond} is maintained during the recursive calls. In other words, since $\AllMakeIt$ can only be used on inputs $(m,U)$ satisfying $u_b \leq T(m,U)$ we need to demonstrate that the inputs used in the recursive calls do, indeed, satisfy this condition. Formally, we need to show that $u_k \leq T(m_k,U_k)$, $k=1,\ldots,b-1$. The following lemma demonstrates this.
    \begin{lemma}\label{cor:fast_set}
        If $u_b \leq T(m,U)$ then $u_{k} \leq T(m_{k},U_{k})$ for $k=1,\ldots,b$.
    \end{lemma}
    \begin{proof}
        Follows from a simple induction argument on $k=b,b-1,\ldots,1$. The base case $k=b$ is true by assumption. Assuming that the lemma holds for some $k \le b$ we have $u_k \leq T(m_k,U_k)$ which implies that $T(m_k,U_k) \leq T(m_{k-1},U_{k-1})$ by part~\ref{lempart:TmU_order} of Lemma~\ref{lem:lbs-for-bsprob}. Since $u_{k-1} \le u_k$ the inductive step is completed.
    \end{proof}
            
    The difficulty in finishing the proof of correctness of \textsc{AllMakeIt} lies in the fact that we do not have a closed-form expression for $X$ and $M$ that are returned by the algorithm. Such closed-form expressions are hard to obtain due to the recursive structure of the algorithm and unspecified \textsc{SolveLP} procedure. In turn, this makes it difficult to verify that $(X,M)$ satisfy the requirements of Definition~\ref{def:feasible}. To sidestep these issues, instead of working with partition $X$ directly, we shall work with an \emph{unexpanded partition} $Z$  and an \emph{unexpanded schedule matrix} $M'$, which we define next.

    Figure~\ref{fig:all_make_it} refers to the unexpanded partition of the interval $[0,1]$ into $m$ subintervals $Z = [z_1,z_2,\ldots,z_m]$. Intervals $z_1, \ldots, z_{m-b}$ describe the first phase of the algorithm and intervals $z_{m-b+1}, \ldots, z_m$ refer to the second phase of the algorithm. We want to be able to refer to the intervals in the second phase with simple indices, so we introduce notation $c(k) = m-b+1+k$ for $k = 0, \ldots, b-1$. This way, the second phase has intervals $z_{c(0)}, z_{c(1)}, \ldots, z_{c(b-1)}$. The \emph{unexpanded matrix} $\widetilde{M'}$ is defined in terms of inverse speeds as follows. For the first phase of the algorithm, the matrix $\widetilde{M'}$ is identical to $\widetilde{M}$, that is, for $1 \leq j \leq m_0$ we have
    \[\widetilde{M'}(i,j) = \begin{cases}
        u_{i-j+1},& j \leq i \leq b+j-1\\
        1,& \mbox{otherwise}
    \end{cases}\]
    and for the second phase, that is, for $0 \le k \le b-1$ we have
    \[\widetilde{M'}(i,c(k)) = \begin{cases}
        T(m_k,U_k),& 1 \leq i \leq m_k\\
        u_{k+i-m_k},& m_k < i \le m
    \end{cases}.\]
    In words, the unexpanded matrix is obtained by treating agents $1\le i \le m_k$ during intervals $z_{c(k)}$ as if they were each riding a bike with speed $1/T(m_k,U_k)$. Of course, this is justified provided that \textsc{AllMakeIt}$(m_k, U_k)$ produces a feasible schedule with all agents having the same arrival time $T(m_k, U_k)$, but this can be assumed to hold in an inductive argument.

    We say that an unexpanded schedule partition $Z$ is \emph{valid} with respect to the unexpanded matrix $M'$ if the following condition holds: for $j \in [m]$ the agents $1,2,\ldots,j$ all reach the end of $z_j$ at the same time in the unexpanded schedule. In particular, validity implies that all agents have the same arrival time in $(Z, M')$. We are now ready to state the main result of this section -- the expanded schedule produced by \textsc{AllMakeIt} is optimal under the condition $u_b \le T(m,U)$.

    \begin{theorem}\label{thm:all_make_it}
        Consider an input $(m,U=\{u_1 \le u_2 \le \cdots \le u_b\})$ to the \BS problem satisfying $u_b \leq T(m,U)$. Let $(X,M)$ be the schedule produced by \textsc{AllMakeIt}$(m,U)$. The schedule $(X,M)$ is feasible and has completion time $\tau(X,M) = T(m,U)$.
    \end{theorem}

    The (rather technical) proof proceeds along the following lines: (1) we make sure that the recursive calls made in the \textsc{AllMakeIt} procedure satisfy the precondition $u_b \le T(m,U)$, (2) we prove that there is a partition that makes schedule feasible and results in completion time $T(m,U)$. Thus, by Lemma~\ref{lm:semi_opt} such a partition can be found by solving a related LP. The proof of step (2) is by induction on the number of bikes $b$ and the difficulty is that we don't have a closed-form expression of the partition so we cannot easily verify property 3 of Definition~\ref{def:feasible}. We sidestep this difficulty by working with so-called ``unexpanded'' schedule (unexpanded partition and matrix) that precisely captures the recursive structure of Figure~\ref{fig:all_make_it}. We show how an unexpanded schedule can be transformed into a complete feasible schedule.

    \begin{proof}(Theorem~\ref{thm:all_make_it})
        First, observe that in light of part~\ref{lempart:same-arrival} of Lemma~\ref{lem:lbs-for-bsprob} the condition $\tau(X,M) = T(m,U)$ is equivalent to all agents having the same arrival time in $(X,M)$. We will use this fact multiple times throughout this proof. 
        
        We prove the theorem by induction on $b$. For the base case $b= 0$, observe that $T(m, \emptyset) = 1$ and in the schedule produced by \textsc{AllMakeIt}$(m, \emptyset)$ all agents walk. So the schedule is feasible and has completion time $1$.
        Now consider some $b \ge 1$ and assume the claim holds for all acceptable inputs with at most $b-1$ bikes. Consider an input $(m, U = \{u_1, \ldots, u_b\})$ and let $(X,M) = \textsc{AllMakeIt}(m, U)$. 
    
        Let $(X_k, M_k) = \textsc{AllMakeIt}(m_k, U_k)$. These are all the recursive calls made by the algorithm to populate the schedule matrix during the intervals $z_{c(k)}$ for $k = 0, \ldots, b-1$. By Lemma~\ref{cor:fast_set}, the inductive assumption applies to $(X_k, M_k)$ and we conclude that all these schedules are feasible and complete in optimal time. Moreover, in each $(X_k, M_k)$ all agents have the same arrival times.
    
        Let us assume that an unexpanded partition $Z$ that is valid for the unexpanded matrix $M'$ exists and see why this finishes the proof. The partition $Z$ for $(m,U)$ can be expanded into a true partition $X_Z$ that would go along with expanded $M$ via a simple procedure: splice vector $z_{c(k)} X_k$ into the entry at position $c(k)$ of $Z$ for $k = 0, \ldots, b-1$. The correctness of this transformation follows from Claim~\ref{claim:scaling}. It is then easy to see that that $(X_Z, M)$ satisfies the conditions of Definition~\ref{def:feasible}. Requirements $1$ and $2$ are clear from Figure~\ref{fig:all_make_it} and the fact that $(X_k, M_k)$ are all feasible and optimal. To see that requirement $3$ is satisfied observe that in every interval $j \leq i-1$, agent $i-1$ uses a bike at least as fast as agent $i$ does. Moreover, the only pickup-switches that occur during the algorithm (excluding any such bike switches occurring in a recursive call) do so between pairs of agents $i$ and $i-1$ at the end of intervals $z_j$, $j \leq i-1$. In order for these switches to be possible agent $i-1$ must reach the end of $z_j$ no later than agent $i$ does. This, of course, is guaranteed by the fact that agent $i-1$ uses a bike at least as fast as agent $i$ in every interval preceding the interval $z_j$. We can thus conclude that all pickup-switches are possible.
    
        Thus, by showing the existence of $Z$ we ultimately show that $X_Z$ exists that makes $(X_Z,M)$ feasible and every agent arrives at the end at the same time (guaranteed by validity of $Z$), meaning that $\tau(X_Z, M) = T(m, U)$. In turn, this implies that \textsc{SolveLP} procedure would find a partition $X$ with the same properties as $X_Z$ in Algorithm~\ref{alg:allmakeit}.
    
        By Claim~\ref{claim:scaling} and the fact that rescaling does not alter the validity of a partition, it suffices to construct a valid partition over any interval $[0,a]$ for $a > 0$. To that end, we let $z_1 = 1$ and for $j = 2, \ldots, m$ we let
        \[z_{j} = \frac{\sum_{p=1}^{j-1} [\widetilde{M'}(j,p)- \widetilde{M'}(j-1,p)] z_p}{\widetilde{M'}(j-1,j) - \widetilde{M'}(j,j)}.\] 
        First, we show that $z_j \ge 0$ for all $j \in [m]$ implying that this $Z$ is, indeed, a partition over the interval $[0, \sum_{j=1}^m z_j]$. We do so by demonstrating that each coefficient of $z_p$ is non-negative, so the non-negativity of $z_j$ follows by a simple induction on $j$.
    
        Consider expressions appearing as coefficients in front of $z_p$ above. We have already indicated that agent $j-1$ uses a bike at least as fast as agent $j$ in every interval preceding $z_{j}$ and thus we indeed have that $\widetilde{M'}(j-1,p)  \le \widetilde{M'}(j,p)$ for each $p=1,2,\ldots,j-1$. If $2 \le j \leq m-b$ then $\widetilde{M'}(j,j) = u_1$, $\widetilde{M'}(j-1,j)=1$, and so $\widetilde{M'}(j-1,j) \geq \widetilde{M'}(j,j)$ is certainly satisfied for $j=2,\ldots,m-b$. Moreover, as $k$ runs from $0$ to $b-1$, $c(k)$ runs from $m-b+1$ to $m$. Thus, for $j \ge m-b+1$ we have $\widetilde{M'}(j,j) = \widetilde{M'}(c(k),c(k)) = u_{k+1}$ while $\widetilde{M'}(j-1,j) = \widetilde{M'}(c(k)-1, c(k)) = T(m_k,U_k)$, and the fact that $u_{k+1} \le T(m_k,U_k)$ follows directly from Lemma~\ref{cor:fast_set}. 
    
        Second, we show that $Z$ is valid. Similarly to \eqref{eq:tij_general}, we have
        \[t'_{i,j} = \sum_{p=1}^{j} \widetilde{M'}(i,p) z_p.\]
        Rearranging the definition of $z_j$ we obtain
        \[(\widetilde{M'}(j-1,j) - \widetilde{M'}(j,j))z_{j} = \sum_{p=1}^{j-1} [\widetilde{M'}(j,p)- \widetilde{M'}(j-1,p)] z_p.\]
        Therefore
        \[\sum_{p=1}^{j} \widetilde{M'}(j,p) z_p - \sum_{p=1}^{j} \widetilde{M'}(j-1,p) z_p = 0.\]
        Thus, we have $t'_{j,j} = t'_{j-1,j}$. In words, the $j^\text{th}$ completion time of $j^\text{th}$ agent and the $j^\text{th}$ completion time of $(j-1)^\text{st}$ agent are the same. Thus, the validity of $Z$ follows by a simple induction.
    \end{proof}
    
    Unfortunately, the size of the schedule produced by \textsc{AllMakeIt} is exponential in the number of bikes.

    \begin{theorem}\label{thm:exp_size}
        $\textsc{AllMakeIt}(m,U)$ has size $1$ if $b=0$ and size $(m-1)2^{b-1}+1$ if $0 < b < m$.
    \end{theorem}
    \begin{proof}(Theorem~\ref{thm:exp_size})
        Let $n_k$ represent the size of $\textsc{AllMakeIt}(m_k,U_k)$. Referring to Figure~\ref{fig:all_make_it} one can observe that the size $n_b$ of $\textsc{AllMakeIt}(m,U)$ satisfies the following recursion
        $n_b = m - b + 1 + \sum_{j=1}^{b-1} n_j.$
        We will now proceed to prove the lemma by induction on $b$. 
        For $b=1$, we have $n_1 = m-1+1 = m$ as can be easily verified by consulting Figure~\ref{fig:all_make_it} for the case $b=1$. Thus, the base case holds. 
        Next, consider some $b \ge 2$, and assume the statement holds for all $k \leq b-1$. Then, we have
        \[n_b = m-b + 1 + \sum_{k=1}^{b-1} (m-1)2^{k-1}+1 = m + (m-1)(2^{b-1}-1) = (m-1)2^{b-1}+1\]
        which completes the proof by induction. 
    \end{proof}

\subsection{\textsc{Standardize} and \textsc{Reduce} Procedures}
    In this section, we present two procedures called \textsc{Standardize} and \textsc{Reduce} that would allow us to overcome the main drawback of \textsc{AllMakeIt} algorithm -- the exponential size of a schedule it outputs. 
    We begin with the \textsc{Standardize} procedure.

    Consider a feasible schedule $(X,M)$ of size $n$. If there exists an entry $x_j = 0$ in $X$ then this entry and the $j^\text{th}$ column of $M$ do not contain any useful information. These so-called {\em zero columns} only serve to increase the size of $(X,M)$ and can be removed from the algorithm by deleting $x_j$ and the corresponding column of $M$. Similarly, whenever we have two subsequent columns $j-1$ and $j$ of $M$ that are identical (i.e. none of the agents switch bikes between the intervals $x_{j-1}$ and $x_j$) then the $j^\text{th}$ column is a {\em redundant column} and can be removed from the algorithm by deleting the $j^\text{th}$ column of $M$ and merging the intervals $x_{j-1}$ and $x_j$ in $X$.

    In addition to removing these useless columns, we can also remove from the algorithm any useless switches that might occur. To be specific, if ever it happens during $(X,M)$ that an agent $i$ performs a pickup-switch with an agent $i'$ at the same position and {\em time}, then this bike switch can be avoided. Indeed, assume that $M$ instructs the agents $i$ and $i'$ to switch bikes at the end of some interval $x_k$ and assume that this switch happens at the same time. Then we can remove this switch by swapping the schedules of agents $i$ and $i'$ in every interval succeeding the interval $x_k$. In other words, we need to swap $M(i,j)$ and $M(i',j)$ for each $j = k+1,k+2,\ldots,n$. We will call these types of switches {\em swap-switches}. Note that a swap-switch is also a pickup-switch. Clearly, the removal of swap-switches will not affect the feasibility nor the overall completion time of an algorithm.

    We say that a schedule $(X,M)$ is expressed in \emph{standard form} if $(X,M)$ does not contain any zero columns, redundant columns, or swap-switches.
    In the appendix we present and analyze the procedure $\textsc{Standardize}(X,M)$ that takes as input a feasible schedule $(X,M)$ and outputs a feasible schedule $(X',M')$ in standard form. The main result is:

    \begin{lemma}
        A feasible schedule $(X,M)$ of size $n$ can be converted to standard form in $O(m^2 n)$ computational time using $O(mn)$ space.
    \end{lemma}

    Next, we describe the \textsc{Reduce} procedure that demonstrates that a feasible schedule is never required to have size larger than $m$. This procedure is presented as Algorithm~\ref{alg:full_reduce}. The idea is simple: we start with a feasible schedule $(X_M,M)$ ($X_M$ is the induced partition vector of $M$), apply the standardization procedure to get a new schedule $(X',M')$ of reduced size, and then iterate on the schedule $(X_{M'},M')$ ($X_{M'}$ is the induced partition vector of $M'$). The iteration stops once we reach a schedule $(X_{M''},M'')$ that is already in standard form. The key thing is to demonstrate that the iteration always stops at a schedule with size $n \leq m$ (without increasing the completion time of the schedule). This is the subject of Theorem~\ref{thm:full_reduce}.

        \begin{algorithm}[H] \caption{$\textsc{Reduce}(U,M)$} \label{alg:full_reduce}
            \begin{algorithmic}[1]
            \State $X \leftarrow \textsc{SolveLP}(U,M)$;   \Comment{Solve the LP \eqref{eq:LP2}}
            \Repeat       
            \State $(X,M) \leftarrow \textsc{Standardize}(U,X,M)$;
            \State $X \leftarrow \textsc{SolveLP}(U,M)$;
            \Until {$(X,M)$ is in standard form}
            \State \Return {$M$;}                
            \end{algorithmic}
        \end{algorithm}
        \begin{theorem}\label{thm:full_reduce}
            Consider a feasible schedule $(X_M,M)$ of size $n > m$ that completes in time $\tau$. Then $\textsc{Reduce}(M)$ produces a feasible schedule $(X_{M'},M')$ with size $n' \leq m$ that completes in time $\tau' \leq \tau$. The computational complexity of $\textsc{Reduce}(U,M)$ is $O(\poly(m,n))$ time and space.
        \end{theorem}
        \begin{proof}
            Consider a feasible schedule $(X,M)$ of size $n>m$ that completes in time $\tau$. We may assume without loss of generality that $X=X_M$ is an induced partition vector of $M$. Then $X_M$ is an optimal vertex of the LP \eqref{eq:LP2}.
        
            As one can see we have $n+1$ variables in the LP \eqref{eq:LP2}, one equality constraint, and a number of inequality constraints. We know that there is a feasible solution and thus the equality constraint is certainly satisfied leaving us with $n$ degrees of freedom. Of the remaining constraints there are: $m$ of the form $\tau \geq t_i$; $n$ of the form $x_j \geq 0$; and some number $s$ (equal to the number of pickup-switches in $M$) of the form $t_{i',j-1} \leq t_{i,j-1}$. Of these total $m+n+s$ inequality constraints, $n$ of them will be made tight by the optimal solution $X_M$. Since $n > m$ at least $n-m>0$ of the inequalities $x_j \geq 0$ and/or $t_{i',j-1} \leq t_{i,j-1}$ will be tightly satisfied. Note that this implies that $(X,M)$ cannot be expressed in standard form since each tight inequality $x_j\geq 0$ corresponds to a zero column and each tight inequality $t_{i',j-1} \leq t_{i,j-1}$ corresponds to a swap-switch. Assume that $n_0 \geq 0$ of the constraints $x_j\geq 0$ are tight and $n_s \geq 0$ of the constraints $t_{i',j-1} \leq t_{i,j-1}$ are tight. Note that we must have $n_0 + n_s \geq n-m > 0$.
            
            Let $(X',M')$ represent the schedule obtained by converting $(X,M)$ into standard form. Then $(X',M')$ has size $n' \leq n-n_0$. If $n' \leq m$ then we are done since $(X',M')$ will complete in the same time as $(X,M)$ (the completion times are unaffected by the standardization procedure). If $n' > m$ then we consider the schedule $(X_{M'},M')$, also of size $n'$, and where $X_{M'}$ is the induced partition vector of $M'$, i.e. $X_{M'}$ is an optimal solution to the same LP \eqref{eq:LP2}. Clearly the completion time of $(X_{M'},M')$ will be at least as good as $\tau$. 
            
            Now observe that this new LP has $n'+1$ variables, $1$ equality constraint, and $m+n'+s'$ inequality constraints: $n' \leq n-n_0$ of the form $X'(j) \geq 0$, and $s' \leq s-n_s$ of the form $t'_{'i',j-1} \leq t'_{i,j-1}$. Thus, either $(X_{M'},M')$ has size $n' < n$ or $X_{M'}$ optimally solves an LP involving fewer constraints than \eqref{eq:LP2}. In the latter case, since $n' > m$, at least one of the constraints of the form $x_j \geq 0$ or $t'_{S'(i,j),j-1} \leq t'_{i,j-1}$ will be tightly satisfied by $(X',M')$ implying that $(X',M')$ is not in standard form. Thus, we can repeat this whole procedure multiple times, each time reducing the size of the schedule and/or reducing the number of inequalities in the corresponding LP. Either way, the size of the schedule must eventually decrease to $m$.
        
            To deduce the computational complexity of the reduce procedure we observe that here can be at most $m(n-1)$ switches in any schedule of size $n$. Thus, the LP \eqref{eq:LP2} is an LP with $O(n)$ variables and $O(nm)$ constraints and can be solved using $O(\poly(n,m))$ steps. Since the standardization procedure is also $O(\poly(n,m))$ each iteration of Algorithm~\ref{alg:full_reduce} runs in $O(\poly(n,m))$. In each iteration the size is decreased by one and/or at least one constraint is removed. Thus, there can be at most $O(nm)$ iterations and we can conclude that the overall computational complexity of the reduction procedure is also $O(\poly(n,m))$.     
        \end{proof}        
        We say that a schedule $(X_M,M)$ is in {\em reduced} form if $M$ is obtained as an output of Algorithm~\ref{alg:full_reduce}. Clearly, any schedule that is in reduced form is also in standard form.

\subsection{\texorpdfstring{$\textsc{AllMakeIt}^*$}{AllMakeIt*}: Computationally Efficient Version of \textsc{AllMakeIt}}   
    There are two obstacles we need to overcome in order to make $\AllMakeIt$ algorithm run in polynomial time. The first obstacle is information-theoretic: the size of the schedule produced by $\AllMakeIt$ algorithm is exponential in the number of bikes. The second obstacle is computational: the sheer number of recursive calls made by $\AllMakeIt$ algorithm is exponential in the number of bikes, so even if sizes of all matrices were $1$ it still would not run in polynomial time. The \textsc{Reduce} procedure allows us to overcome the first obstacle. We overcome the second obstacle by observing that the number of \emph{distinct} sub-problems in $\AllMakeIt$ algorithm and all of its recursive calls is at most $b$. Thus, we can replace recursion with dynamic programming (DP) to turn $\AllMakeIt$ into a polynomial time algorithm $\AllMakeIt^*$. A more thorough explanation/analysis of $\AllMakeIt^*$ algorithm follows.

        We begin by establishing the following recursive invariant:

        \begin{lemma}\label{lem:allmakeit_invariant}
            Let $(m,U=\{u_1, \ldots, u_b\})$ be an input satisfying $u_b \le T(m,U)$. Let $(m',U')$ be an input to a recursive call that is made at any point during the execution of $\AllMakeIt$ algorithm on input $(m,U)$. Then
            \begin{enumerate}
                \item $m'-|U'| = m-|U| = m-b$, and
                \item $U' = \{u_1, \ldots, u_k\}$ for some $k \in \{0, 1, \ldots, b-1\}$.
            \end{enumerate}
        \end{lemma}
        \begin{proof}
            We prove this lemma by induction on $b$. The base case of $b = 0$ is trivially true.

            For the inductive step, consider $b \ge 1$. When the algorithm $\AllMakeIt$ is executed on $(m,U)$ it makes calls to $\AllMakeIt(m_k, U_k)$ for $k \in \{0, \ldots, b-1\}$. Since $m_k - |U_k| = (m-b+k) - k = m-b$, the statement holds for all recursive calls done at this level. Now consider input $(m',U')$ in a recursive call made within $\AllMakeIt(m_k, U_k)$. Since $|U_k| \le b-1$ we can apply the induction assumption and conclude that $m'-|U'| = m_k - |U_k|$ and that $U' = \{u_1, \ldots, u_{k'}\}$ for some $k' \in \{0, 1, \ldots, k-1\}$. Since $m_k - |U_k| = m-b$ and $k'-1 \le b-2$, the inductive step holds.
        \end{proof}

        The above lemma implies that all recursive sub-problems can be parameterized by a single parameter $k$ -- number of bikes. The number of agents associated with this sub-problem is fixed by the identity $m'-|U'| = m-b$, since $|U'| = k$ implies $m' = m-b+k$. Thus, we introduce the dynamic programming table $DP[\cdot]$ with the intended meaning
        \[DP[k] = \text{optimal schedule matrix with } m-b+k \text{ agents and bikes } \{u_1, \ldots, u_k\}.\]
        By filling the $DP[k]$ table in order $k = 0, k=1, \cdots, k = b$ we guarantee that solutions to all subproblems used in $DP[k]$ are available at the time we fill in $DP[k]$. The \textsc{Reduce} procedure is then used to make sure that the size of $DP[k]$ schedule is at most $m-b+k$. We call the resulting algorithm $\AllMakeIt^*$ and the pseudocode for the $\AllMakeIt^*$ procedure is given in Algorithm~\ref{alg:allmakeitstar} .
        
        \begin{algorithm}[H] \caption{$\textsc{AllMakeIt}^*(m,U)$} \label{alg:allmakeitstar}
            \begin{algorithmic}[1]             
            \Input {$m$ -- number of agents, $U = \{u_1 \leq u_2 \leq ... \leq u_b\}$ -- inverse bike speeds;\\
            \textbf{Precondition:} $b=0$ or $u_b \le T(m,U)$}
            \Init
            \State $DP[0...b] \leftarrow$ empty table with $b+1$ entries; \Comment{DP table that stores schedules}
            \Begin
            \State $DP[0] \leftarrow [0;0;...;0]$; \Comment{Base case: no bikes, schedule is trivial, every agent walks}
            \For {$k=1$ to $b$}
                \State $M$ initialized to an empty $(m-b+k) \times (m-b)$ matrix;
                \For {$j=1$ to $m-b$} \Comment{Phase 1}
                    \For {$i=1$ to $m-b+k$}
                        \If {$j \leq i \leq j+k-1$}
                            \State $M(i,j) \leftarrow i-j+1$;
                        \Else
                            \State $M(i,j) \leftarrow 0$;
                        \EndIf
                    \EndFor
                \EndFor
                \For {$j=m-b+1$ to $m-b+k$} \Comment{Phase 2}
                    \State $M' \leftarrow DP[j-(m-b+1)]$;
                    \State $n' \leftarrow \textsc{Size}(M')$; \Comment{Get size of $M'$}
                    \For {$i=j$ to $m-b+k$}
                        \State $M' \leftarrow [M'; [i-m+b]_{n'}]$; \Comment{Append to $M'$ a row with $n'$ entries equal to $i-m+b$}
                    \EndFor
                    \State $M \leftarrow [M, M']$; \Comment{Append $M'$ to the right of $M$}
                \EndFor
                \State $DP[k] \gets \textsc{Reduce}(U_k,M)$; \Comment{Prevent size from blowing up} 
            \EndFor
            \State \Return $DP[b]$;
            \end{algorithmic}
        \end{algorithm} 

    \begin{theorem} \label{thm:bs_ub_ge_tmu}
        Let $(m,U)$ be the input to the \BS problem such that $u_b \ge T(m,U)$. The algorithm $\AllMakeIt^*$ runs in polynomial time on input $(m,U)$ and returns an optimal schedule $(X,M)$ with $\tau(X,M) = T(m,U)$.
    \end{theorem}    
    
    \begin{proof}
        The correctness of the procedure follows from the correctness of $\AllMakeIt$ procedure, Lemma~\ref{lem:allmakeit_invariant} and Figure~\ref{fig:all_make_it}. Thus, it is only left to prove that the procedure runs in polynomial time.

        The pseudocode consists of several nested for loops, each with at most $m$ iterations. Within the $k^\text{th}$ iteration of the outer for loop, the code performs polynomially many concatenation operations with matrices $DP[j]$ for $j \le k-1$ and applies a polynomial time \textsc{Reduce} procedure. As long as the size of the largest matrix ever created in $\AllMakeIt$ is polynomial in $m$, the entire code runs in polynomial time. The schedule matrix $M_k$ created during the $k^\text{th}$ iteration  is largest immediately prior to calling \textsc{Reduce}. Since \textsc{Reduce} operation guarantees $size(DP[j]) \le m_j$ for $j \le k-1$, we have
        \[size(M_k) = m-b+1 + \sum_{j=1}^{k-1} size(DP[j]) \le m-b+1 + \sum_{j=1}^{k-1} m_j \le m + m(b-1) = mb.\]
        This means that all matrices encountered during the execution of $\AllMakeIt^*$ are of polynomial size.
    \end{proof}

\subsection{Finding an Optimal Schedule for the Case \texorpdfstring{$u_b > T(m,U)$}{ub > T(m,U)}}\label{ssec:ub_g_T}

    In this section we solve the \BS problem efficiently and optimally for the case when the slowest bike is the bottleneck. The idea is to reduce it to the case of Subsection~\ref{ssec:all_make_it}. 

    We make the following observation:
    \begin{lemma} \label{lem:ub_gt_tmu}
        Let $(m, U=\{u_1 \le u_2 \ldots \le u_b\})$ be the input to the \BS problem. If $u_b > T(m,U)$ then there exists $k \in \{1, 2, \ldots, b-1\}$ such that
        $u_{k} \le T(m_k, U_{k}) \le u_b$.
    \end{lemma}
    \begin{proof}
        We give a proof by induction on $b$. For $b = 1$, the claim is trivially true, since its premise is false: $u_1 > T(m, \{u_1\}) = 1 - \frac{1}{m} (1-u_1)$ implies that $u_1 > 1$, a contradiction.

        Next consider $b \le 2$ and assume the claim is true for any input with $b-1$ bikes. We know that $u_{b-1} \le u_b$. If $u_{b-1} \le T(m-1, U_{b-1})$ then Part~\ref{lempart:TmU_order} of Lemma~\ref{lem:lbs-for-bsprob} implies that $T(m-1, U_{b-1}) \le T(m, U)$ which in turn is $<u_b$ by assumption. So  $k=b-1$ satisfies $u_{k} \le T(m_k, U_{k}) \le u_b$. If instead $u_{b-1} >  T(m-1, U_{b-1})$, we can apply the induction assumption to $(m-1, U_{b-1})$ and obtain a $k' \in \{1, 2, \ldots, b-2\}$ such that $u_{k'} \le T(m_{k'}, U_{k'}) \le u_{b-1}$. We conclude that 
        \[u_{k'} \le T(m_{k'}, U_{k'}) \le u_{b-1} \le u_b,\]
        which completes the proof. 
    \end{proof}

    With this lemma we can prove the main result of this section. 

    \begin{theorem}\label{thm:bs_ub_le_tmu}
        Let $(m,U)$ be the input to the \BS problem such that $u_b > T(m,U)$. There exists a polynomial time algorithm that constructs an optimal schedule $(X,M)$ with $\tau(X,M) = u_b$.
    \end{theorem}

    \begin{proof}
        First, we find $k$ such that $u_{b-k} \le T(m-k, U_{b-k}) \le u_b$, which is guaranteed to exist by Lemma~\ref{lem:ub_gt_tmu}. Second, we apply the algorithm $\AllMakeIt^*(m-k, U_{b-k})$ to obtain a schedule $(X_1, M_1)$ in which $m-k$ agents share bikes $U_{b-k}$ to complete the interval $[0,1]$ in time $\tau(X_1, M_1) = T(m-k, U_{b-k}) \le u_b$. Next, we create a schedule $(X_2, M_2) = ([1], [b-k+1; b-k+2; \ldots; b])$. In words, in $(X_2, M_2)$ the remaining $k$ agents use the remaining bikes $b-k+1, b-k+2, \ldots, b$ without sharing, i.e., each agent gets assigned a unique bike and the agent rides the bike for the entirety of the interval $[0,1]$. Notice that the agents in $(X_2, M_2)$ have completion times $u_{b-k+1} \le u_{b-k+2} \le \ldots \le u_b.$ Lastly, we let $(X,M)$ be the concatenation of schedules $(X_1, M_1)$ and $(X_2, M_2)$. This means that the partition $X_2$ needs to be refined to match the partition $X_1$ and the corresponding entries of $M_2$ need to be duplicated, and then $(X_1, M_1)$ can be simply stacked on top of $(X_2, M_2)$. The feasibility of $(X,M)$ is evident and $\tau(X,M) \le \max(\tau(X_1, M_1), \tau(X_2, M_2)) = u_b$. The schedule matrix $M$ is schematically shown in Figure~\ref{fig:all_make_it_ublb}. The running time is clearly polynomial, since $\AllMakeIt^*$ runs in polynomial time and the rest of the computation consists of a few simple splicing and concatenation operations on lists of polynomial size. Since the procedure is rather straightforward, we omit its pseudocode. 
        \begin{figure}[!h]
            \centering
            \includegraphics[scale=0.45]{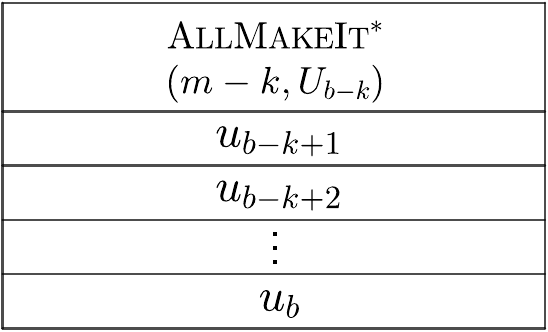}
            \caption{The schedule matrix for the \BSP problem for the case when $u_b > T(m,U)$. Here $k$ is the smallest integer such that $u_{b-k} \le T(m-k, U_{b-k})$.\label{fig:all_make_it_ublb}}
        \end{figure}    
    \end{proof}

\section{The \texorpdfstring{\RBSP}{RBS} problem}\label{sec:rbs}
    In Section~\ref{sec:bs} we presented an optimal algorithm for the \BSP problem. Recall that in the case $u_b > T(m,U)$ the minimum arrival time is $u_b$. Thus, the arrival time in that case is controlled by how much time it takes the slowest bike to travel from $0$ to $1$. This suggests that if we relax the requirement of all bikes making it to the end and allow agents to abandon the slowest bike, for example, the overall completion time for the agents and the remaining bikes might be improved. This naturally leads to the \RBSP problem. We begin by developing a polynomial time algorithm that solves the \RBSP problem optimally when the abandonment limit $\ell$ is $1$, that is we allow at most one bike to be abandoned.

    \emph{High-level overview of the algorithm.} Our approach to the \RBSP problem with the abandonment limit $\ell = 1$ mimics to some extent our approach to the \BSP problem. The structure of the optimal schedule depends on the relationships between speeds of two slowest bikes and certain expressions lower bounding the optimal completion time. More specifically, 
    we design an algorithm \textsc{AllButOne} that solves the \RBSP problem with the abandonment limit $\ell=1$ optimally in polynomial time. We follow the next steps:

    \begin{enumerate}
        \item We begin by generalizing the lower bound of Lemma~\ref{lem:lbs-for-bsprob} to the situation where multiple bikes may be abandoned.
        \item Analyzing the lower bound from the first step we can immediately conclude that $\AllMakeIt^*$ algorithm provides an optimal solution to the \RBSP problem (in fact for any value of $\ell$) under the condition $u_b \le T(m,U)$. This is the content of of Theorem~\ref{thm:rbs_ub_lt_tmu}, which is also subsumed by Theorem~\ref{thm:rbs-spec-case1} below.
        \item By the previous step, it remains to handle the case $u_b > T(m,U)$. We split this case into two sub-cases depending on the relationship between the inverse speed of the second slowest bike $u_{b-1}$ and $T_1$, which is the lower bound from Step 1 specialized to the case $\ell = 1$ (the subscript in $T_1$ indicates the abandonment limit). We handle the case of $u_{b-1} \le T_1(m,U)$ first.
        \item Lastly, we show how to handle the remaining sub-case of $u_b > T(m,U)$, namely, when $u_{b-1} > T_1(m,U)$.
    \end{enumerate}

    Carefully carrying out the above steps results in the following two theorems that jointly provide a solution to the \RBSP problem with the abandonment limit $\ell =1$.
    \begin{theorem}
        \label{thm:rbs_ub_gt_tmu_ub1_lt_t1mu}
            Let $(m,U = \{u_1 \le u_2 \le \cdots \le u_b\})$ be such that $u_b > T(m,U)$ and $u_{b-1} \le T_1(m,U)$. The schedule $(X,M)$ output by $\textsc{AllButOne}(m,U)$ is feasible and has completion time $\tau(X,M) = T_1(m,U)$. In particular, $(X,M)$ is an optimal solution to the \RBSP problem with abandonment limit 1. 
        \end{theorem}
        \begin{theorem}\label{thm:rbs_ub_gt_tmu_ub1_gt_t1mu}
        Let $(m,U = \{u_1 \le u_2 \le \cdots \le u_b\})$ be such that $u_b > T(m,U)$ and $u_{b-1} > T_1(m,U)$. There is a polynomial time computable feasible schedule $(X,M)$ in which at most one bike is abandoned and $\tau(X,M) = u_{b-1}$. In particular, $(X,M)$ is an optimal solution to the \RBSP  problem with abandonment limit 1.
    \end{theorem}

    Having a closer look at our results for the \BSP and \RBSP problems, it turns out that we can solve the \RBSP problem optimally under some conditions. In particular,  
    when all bikes are "fast enough", that is, when $u_b \leq T(m, U)$, we have: 

    \begin{theorem}\label{thm:rbs-spec-case1}
        $\AllMakeIt^*$ solves the \RBSP problem optimally when $u_b \leq T(m, U)$. 
    \end{theorem}

    In fact, even if {\em all but} the slowest bike are fast enough, we can solve the \RBSP problem optimally. In particular, if $u_b > T(m,U)$ and $u_{b-1} \le T_1(m,U)$ then $\AllButOne$ produces a schedule where only a single bike is abandoned. This schedule is also a possible solution for the \RBSP problem for any $\ell \ge 1$. In the following theorem we demonstrate that there is no better schedule.

    \begin{theorem}\label{thm:rbs_ub_lt_t1mu}
        $\AllButOne$ solves the \RBSP problem optimally when $u_b > T(m,U)$ and $u_{b-1} \leq T_1(m,U)$.
    \end{theorem}

    We now present the details of these theorems.

    Many of the techniques that we developed in Section~\ref{sec:bs} for the Bike Sharing problem either can be immediately applied to the Relaxed Bike Sharing problem or they can be applied with minor modifications. Examples of techniques of the former kind are \textsc{Standardize} and \textsc{Reduce} procedures which were explicitly stated to work for arbitrary feasible schedules. Transforming an exponential time algorithm based on a recursive definition of the schedule matrix into a polynomial time algorithm is an example of a technique of the latter kind. In this section, we assume that the reader is familiar with these techniques and so we concentrate only on novel aspects of the \RBSP problem with the abandonment limit $\ell=1$ as compared with the vanilla \BSP problem. In particular, we shall give recursive definitions of schedule matrices and prove their optimality without specifying an algorithm that would compute those matrices. It can be verified that such an algorithm running in polynomial time can be obtained by modifications of our approach from Section~\ref{sec:bs}. We also assume that the reader is familiar with the notation and terminology introduced in Section~\ref{sec:bs}, e.g., $m_k = m-b+k$, $U_k = \{u_1, \ldots, u_k\}$, expanded vs. unexpanded matrix, and so on.

\subsection{Generalized Lower Bound for the \texorpdfstring{\RBSP}{RBS} problem}
    
        In this subsection, we prove a generalization of Lemma~\ref{lem:lbs-for-bsprob} that establishes a couple of lower bounds on the optimal completion time for the \RBSP problem. These lower bounds do not explicitly refer to the abandonment limit $\ell$ because they are sufficiently general that they can be adapted for any value of the abandonment limit $\ell$. To state the lemma, we introduce a new definition:
        \begin{definition}
            Let $(X,M)$ be an arbitrary feasible schedule for the \RBSP problem with inputs $m, b,$ and $U$. Let $y_1, \ldots, y_b \in [0,1]$ be the total distance for which bikes $1, \ldots, b$ are used respectively, in $(X,M)$. Observe that $y_i < 1$ means that bike $i$ is abandoned at position $y_i$, and $y_i = 1$ means the bike makes it to the end. The \emph{abandonment vector} of $(X,M)$, denoted by $Y(X,M)$, is defined as  $Y(X,M)=[y_1, \ldots, y_b]$. When the schedule $(X,M)$ is clear from the context, we simply write $Y$ to denote the abandonment vector.
        \end{definition}
        We are now ready to state and prove generalized lower bounds on the completion time.
    
        \begin{lemma}\label{lem:lbs-for-rbsprob}
            Let $(X,M)$ be an arbitrary feasible schedule for the \RBSP problem with inputs $m, b,$ and $U$. Let $Y$ be the abandonment vector of $(X,M)$. Define
            \begin{equation}\label{eq:arrival-time-Y}
                T(m,U,Y) := 1 - \frac{1}{m} \sum_{k=1}^b (1-u_k) y_k.
            \end{equation}
            Then we have
            \begin{enumerate}
                \item $\tau(X,M) \ge y_i u_i  + (1-y_i) u_1$,\label{lempart:lb_rbs_1}
                \item $\tau(X,M) \ge T(m,U,Y)$,\label{lempart:lb_rbs_2}
                \item $\tau(X,M) = T(m,U,Y)$ if and only if all agents have the same arrival time in $(X,M)$,\label{lempart:tmuy_equals}
                \item $u_1 \leq T(m,U,Y)$.\label{lempart:min_u1_tmuy}
            \end{enumerate}
        \end{lemma}
        \begin{proof}
            \begin{enumerate}
                \item Consider bike $i$. It takes at least $y_i u_i$ time for bike $i$ to reach location $y_i$. If $y_i < 1$, the agent responsible for abandoning this bike is present at location $y_i$ at time $t \ge y_i u_i$. The fastest possible way for the agent complete the journey is if the agent picks up the fastest bike $1$ at $y_i$ and rides it to the end of the $[0,1]$ interval -- this takes additional $(1-y_i)u_1$ time.
    
                \item The total distance travelled by all agents in $(X,M)$ is $m$, since each agent travels distance $1$. The total distance travelled on  bike $k$ is exactly $y_k$, which implies that the sum of distances walked by all agents is $m-\sum_{k=1}^b y_k$. Since walking speed is $1$, the total time agents spend walking is $m-\sum_{k=1}^b y_k$, whereas the total time agents spend on bikes is $\sum_{k=1}^b u_k y_k$. Combining these observations, we obtain:
                \begin{equation}\label{eq:total-time}
                \sum_{i=1}^m t_i(X,M) = \left( m-\sum_{k=1}^b y_k \right) + \sum_{k=1}^b u_k y_k = m - \sum_{k=1}^b (1-u_k) y_k.
                \end{equation}
                Since the maximum is at least the average, we get
                \[ \tau(X,M) \ge \frac{1}{m} \sum_{i=1}^m t_i(X,M) = 1 - \frac{1}{m} \sum_{k=1}^b (1-u_k)y_k = T(m,U,Y).\]
                \item Follows from the fact that the maximum $t_i$ is exactly equal to the the average if and only if all values are the same.
                \item Follows from the fact that the minimum $t_i \geq u_1$ and the average is at least the minimum.
            \end{enumerate}
        \end{proof}
    
        As mentioned earlier, the lower bound $T(m, U, Y)$ can be specialized to any value of the abandonment limit $\ell$ by adding a constraint that $| \{ y_i < 1 \} | \le \ell$. The following definition can be thought of as a specialization of the generalized lower bound from Lemma~\ref{lem:lbs-for-rbsprob} to the case of the \RBSP problem when at most one bike can be abandoned.
        \begin{equation}\label{eq:T_1(m,U)}
            T_1(m, U) = \min_{0 \leq y \leq 1} \max\left\{T(m,U,[1,\ldots,1,y]),\ y u_b + (1-y)u_1\right\}
        \end{equation}
        First, we observe a few basic facts about the $T_1$ measure.
        \begin{lemma}\label{lm:T_1(m,U)}
            If $u_b > T(m,U)$ then $T_1(m,U) = T(m,U,[1,\ldots,1,y]) = y u_b + (1-y)u_1$ for a certain value of $y \in (0,1)$, otherwise $T_1(m,U) = T(m,U)$.
        \end{lemma}
        \begin{proof}
            First observe that the two terms in the max of expression~\ref{eq:T_1(m,U)} move in different directions as $y$ changes -- $T(m,U,[1,\ldots,1,y])$ decreases with increasing $y$ and the term $y u_b + (1-y)u_1$ increases with increasing $y$. Thus, unless it occurs when $y$ is at an endpoint of $[0,1]$, the minimum will occur at a value $y \in (0,1)$ that makes $T(m,U,[1,\ldots,1,y]) = y u_b + (1-y)u_1$. Since $u_1 < T(m,U)$ for any $(m,U)$ it is easy to see that the minimum never occurs at the left endpoint $y=0$. On the other hand, since $T(m,U,[1,\ldots,1,y])$ decreases to $T(m,U)$ at $y=1$, the minimum only occurs at the right endpoint if $T(m,U,[1,\ldots,1,y=1]) = T(m,U) \geq u_b$.
        \end{proof} 
        Second, we show that $T_1$ is, indeed, a lower bound on the completion time for the \RBSP problem when at most $\ell = 1$ bikes may be abandoned.
        \begin{lemma}
            \label{lem:t1mu_lb}
            Let $(m,U = \{u_1 \le u_2 \le \cdots \le u_b \})$ be an input to the \RBSP problem with an abandonment limit $\ell = 1$ and $u_b > T(m,U)$. Let $(X,M)$ be a feasible schedule that satisfies the abandonment condition. Then 
            $\tau(X,M) \ge T_1(m,U)$.
        \end{lemma}
        \begin{proof}
            Suppose that the bike $b$ makes it to the end. Then $\tau(X,M) \geq u_b$. Since $u_b > T(m,U)$ Lemma~\ref{lm:T_1(m,U)} tells us that for some $y \in (0,1)$ we have $T_1(m,U) = y u_b + (1-y) u_1 < u_b$ and we can conclude that $\tau(X,M) \geq T_1(m,U)$ when the bike $b$ makes it to the end.
    
            Now suppose that bike $b$ gets abandoned. Let $Y$ be the abandonment vector of $(X,M)$. By assumptions we have $Y = [1,\ldots,1,y_b]$, $y_b \le 1$, and by Lemma~\ref{lem:lbs-for-rbsprob} we have
            \begin{enumerate}
            \item $\tau(X,M) \ge T(m,U,Y=[1,\ldots,1,y_b])$,
            \item $\tau(X,M) \ge y_b u_b + (1-y_b)u_1$.
            \end{enumerate}
            Combining these inequalities we get
            \begin{align*}
            \tau(X,M) &\ge \max\left\{T(m,U,[1,\ldots,1,y_b]),\ y_b u_b + (1-y_b)u_1\right\}\nonumber\\
            &\ge \min_{0 \leq y \leq 1} \max\left\{T(m,U,[1,\ldots,1,y]),\ y u_b + (1-y)u_1\right\}
            = T_1(m,U).
            \end{align*}
        \end{proof}
        Related to the above proof is the following useful result which mimics part~\ref{lempart:TmU_order} of Lemma~\ref{lem:lbs-for-bsprob}:
        \begin{lemma}\label{lm:T1mU_order}
            If $u_b > T(m,U)$ then for $k=2,3,\ldots,b-1$ we have $T_1(m,U) \leq T_1(m-1,U\setminus \{u_k\})$ iff $u_k \leq T_1(m,U)$ with equality only when $u_k = T_1(m,U)$.
        \end{lemma}
        \begin{proof}
            Let $g(y) = y u_b + (1-y)u_1$, $f(y) = T(m,U,[1,\ldots,1,y])$ and $f'(y) = T(m-1,U\setminus\{u_k\},[1,\ldots,1,y])$. Then we may write $T_1(m,U) = \min_{0 \leq y \leq 1} \max \{f(y),g(y)\}$ and $T_1(m,U \setminus \{u_k\}) = \min_{0 \leq y \leq 1} \max \{f'(y),g(y)\}$. 
            
            We can rewrite $f'(y)$ to be in terms of $f(y)$ by referring to the definition of $T(m,U,Y)$. We have
            \begin{align*}
                f'(y) &= 1 - \frac{1}{m-1}\left[(1-u_b)y + \sum_{j=1,j\neq k}^{b-1}(1-u_j)\right]
                = 1 - \frac{1}{m-1}\left[(1-u_b)y + \sum_{j=1}^{b-1}(1-u_j) - (1-u_k)\right]\\
                &= 1 - \frac{m}{m-1}[1-f(y)] + \frac{1-u_k}{m-1} = f(y) + \frac{f(y)-u_k}{m-1}.
            \end{align*}
            Thus $T_1(m-1,U\setminus \{u_k\}) = \min_{0 \leq y \leq 1} \max \{g(y),f(y)+\frac{f(y)-u_k}{m-1}\}$.
            
            Now let $y_0 = \arg\min_{0 \leq y \leq 1} \max\{f(y),g(y)\}$. Since we are assuming that $u_b > T(m,U)$ we can conclude from Lemma~\ref{lm:T_1(m,U)} that $T_1(m,U) = f(y_0) = g(y_0)$ and $T_1(m-1,U\setminus\{u_k\})$ is located within the interval with endpoints $g(y_0) = T_1(m,U)$ and $f'(y_0)=f(y_0)+\frac{f(y_0)-u_k}{m-1} = T_1(m,U) + \frac{T_1(m,U)-u_k}{m-1}$. The lemma easily follows from this observation.
        \end{proof} 

\subsection{\texorpdfstring{$\AllMakeIt^*$}{AllMakeIt*} Solves the \texorpdfstring{\RBSP}{RBS} Problem when \texorpdfstring{$u_b \le T(m,U)$}{ub <= T(m,U)}}
    When $u_b \le T(m,U)$ then $\AllMakeIt^*$ produces a schedule where no bikes are abandoned. This schedule is also a possible solution for the \RBSP problem with abandonment limit 1, but we need to verify that there is no better schedule. This is a direct consequence of Lemma~\ref{lem:lbs-for-rbsprob} and the following monotonicity property of $T(m,U,Y)$, which is clear from the definition of $T(m,U,Y)$:

    \begin{claim}\label{claim:lb_monotonicity}
    Consider two abandonment vectors $Y = [y_1, \ldots, y_b]$ and $Y' = [y_1', \ldots, y_b']$ such that $Y' \le Y$. That is for all $i \in [b]$ we have $y_i' \le y_i$. Then
    $T(m,U,Y') \ge T(m,U,Y)$.
    \end{claim}

    Next, we prove the main result of this subsection.
    \begin{theorem}
    \label{thm:rbs_ub_lt_tmu}
    $\AllMakeIt^*$ solves the \RBSP  problem with abandonment limit 1 optimally when $u_b \le T(m,U)$.
    \end{theorem}
    \begin{proof}
    Let $(X,M)$ be the schedule produced by $\AllMakeIt^*$ and let $(X',M')$ be some other feasible schedule for the \RBSP problem. Then we have:
    \[\tau(X',M') \ge T(m, U, Y(X', M')) \ge T(m, U, Y(X,M)) = T(m, U) = \tau(X,M).\]
    The first inequality above is by Lemma~\ref{lem:lbs-for-rbsprob}. The second inequality follows from Claim~\ref{claim:lb_monotonicity} by observing that $\AllMakeIt^*$ produces a schedule in which all bikes make it to the end, i.e., $Y(X,M) = [1, 1, \ldots,1]$ and, consequently, $Y(X', M') \le Y(X,M)$.
    \end{proof}

\subsection{Optimal Schedule  when \texorpdfstring{$u_b > T(m,U)$}{ub>T(m,U)} and \texorpdfstring{$u_{b-1} \le T_1(m,U)$}{ub1 le T1(m,U)}}\label{ssec:rbs_t1mu}
    
    In the \BSP problem the relationship between $u_b$ and $T(m,U)$ determined the structure of an optimal schedule. Analogously, in the case of \RBSP with $\ell = 1$ and $u_b > T(m,U)$, the relationship between $u_{b-1}$ and $T_1(m,U)$ determines the structure of an optimal schedule. In the remainder of this subsection we show that when $u_{b-1} \le T_1(m,U)$ there is, in fact, a schedule that achieves completion time $T_1(m,U)$, which is optimal in light of the above lemma.
    
    In order to design a schedule that has the completion time that matches the bound $T_1(m,U)$ we refer to Lemma~\ref{lm:T_1(m,U)}. Since $T_1(m,U) = u_b y + (1-y) u_1$ for a certain value of $y \in (0,1)$, if a schedule with completion time $T_1(m,U)$ exists, then the bike $u_b$ must be used continuously until location $y$ and the agent that abandons bike $b$ must ride bike $u_1$ from that point onward until the end in that schedule. In principle, bike $b$ may be involved in a number of swap-switches before reaching the location $y$, however, we take the simplest approach of dedicating a single agent (namely, agent $m$) to riding bike $b$ for the entirety of distance $y$. It turns out that this partial schedule can be completed to a feasible and optimal schedule by using an approach similar to phase 2 of \textsc{AllMakeIt} algorithm. Using the initial interval $y$, we get a group of synchronized agents sharing some bikes in front of the rest of the agents using the rest of the bikes. The average speed of the synchronized group is slow, so the rest of the agents start to catch up to the synchronized group and keep ``getting absorbed'' into that group. The schedule is represented schematically in Figure~\ref{fig:all_but_one_make_it}. Let $\textsc{AllButOne}(m,U)$ be an algorithm computing this schedule. As mentioned at the beginning of this section, we shall not provide the details of this algorithm understanding that these details can be easily filled in by following the steps of Section~\ref{sec:bs}.     
                
    More specifically, the schedule $M$ output by \textsc{AllButOne} is as follows. We first determine the largest index $q$, $1 \leq q \leq b$, for which $u_q \leq T(m_q, U_q)$. The interval $[0,1]$ is then partitioned into $b-q+1$ parts $Z=[z_1,z_2,\ldots,z_{b-q+1}]$. During the first interval $z_1$ (the intention is $z_1=y$ where $y$ is as above) the group of agents $1,2,\ldots,m_q$ will perform $\textsc{AllMakeIt}$ using the bikes $U_q$, and the rest of the agents use the remaining bikes $U \setminus U_q$. During the interval $j=2,3,\ldots,b-q+1$ the group of agents $1,2,\ldots,m_{q+j-2}$ will perform $\textsc{AllMakeIt}$ using the bikes $\{u_2,\ldots,u_{q+j-2}\}$; agent $i$, $m_{q+j-1}\leq i \leq m-1$, will use the bike $i-(m-b)$; and agent $m$ will use the bike $u_1$. In light of this, we introduce the notation:
    \[U'_k = U_k \setminus \{u_1\} = \{u_2,u_3,\ldots,u_k\}.\] 
        
    We say that the \emph{unexpanded partition} $Z$ is \emph{valid} for $M$ if it satisfies the requirement that agents $1,2,\ldots,m_{q+j-1}$ all reach the end of $z_j$ at the same time. In particular, all agents complete the schedule $M$ at the same time. We are now in a position to prove Theorem~\ref{thm:rbs_ub_gt_tmu_ub1_lt_t1mu}.    
        
    \begin{figure}[!h]
        \centering
        \includegraphics[scale=0.4]{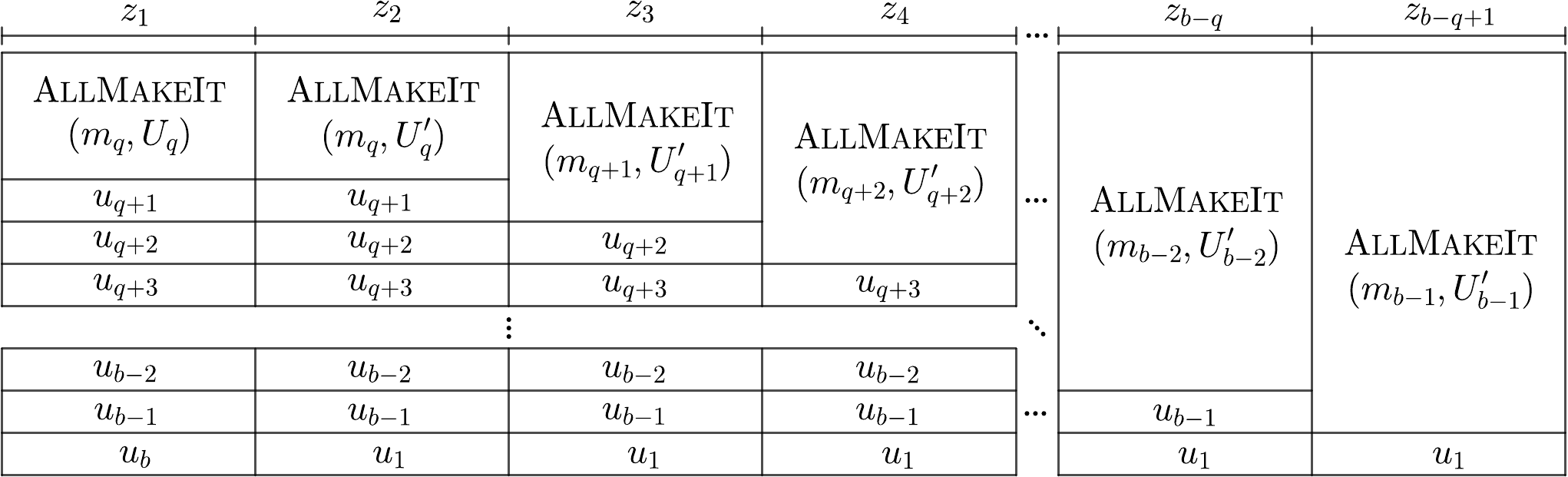}
        \caption{The unexpanded schedule of $\textsc{AllButOne}(m,U)$. The index $q$, $1 \leq q \leq b$, is the largest index for which $u_q \leq T(m_q,U_q)$. In each large rectangle the indicated number of agents will perform $\textsc{AllMakeIt}$ with the indicated set of bikes. With a valid unexpanded partition $Z$, the agents $1,2,\ldots,m_{q+j-1}$ all reach the end of $z_j$ at the same time. \label{fig:all_but_one_make_it}}
    \end{figure}     
        
    \paragraph*{Proof of Theorem~\ref{thm:rbs_ub_gt_tmu_ub1_lt_t1mu}.}
    
    \begin{proof}
        Let $q$ be the largest index with $1 \leq q \leq b$, for which $u_q \leq T(m_q, U_q)$. Observe that $q$ exists by Lemma~\ref{lem:ub_gt_tmu}. Since $\textsc{AllMakeIt}$ can only be used with inputs $(m,U)$ satisfying $u_b \leq T(m,U)$, we must confirm that the inputs used in the calls to $\AllMakeIt$ made during $\AllButOne$ satisfy this condition. Thus, we need to show that $u_q \leq T(m_q,U_q)$ and that $u_{q+i} \leq T(m_{q+i},U'_{q+i})$ for $i=0,1,2,\ldots,b-q-1$. Clearly, the former inequality is true by the definition of $q$. To show the latter inequality, Lemma~\ref{cor:fast_set} tells us that it will suffice to demonstrate only that $u_{b-1} \leq T(m_{b-1}=m-1,U'_{b-1})$. Since we are assuming that $u_{b-1} \leq T_1(m,U)$ we will instead show that $T_1(m,U) \leq T(m-1,U_{b-1}')$. We have
        \[T_1(m,U) \leq T(m,U_{b-1}) \leq T(m-1,U_{b-1} \setminus \{u_1\}) = T(m-1,U'_{b-1})\]
        where the first inequality results from Claim~\ref{claim:lb_monotonicity} and the fact that $T_1(m,U) = T(m,U,Y=[1,\ldots,1,y])$ with $y\in (0,1)$ (see Lemma~\ref{lm:T_1(m,U)}). The second inequality results from part~\ref{lempart:TmU_order} of Lemma~\ref{lem:lbs-for-bsprob}.

        Next we demonstrate that the feasibility requirements of Definition~\ref{def:feasible} are satisfied by the schedule. It is obvious that requirements $1$ and $2$ are satisfied. Moreover, there is only one pickup-switch that occurs outside of the calls to $\AllMakeIt$, and this bike switch occurs at the end of $z_1$ between agent $m$ and one of the agents participating in the call to $\AllMakeIt(m_q,U_q)$. Each of the agents participating in the call to $\AllMakeIt(m_q,U_q)$ will cross $z_1$ with an average speed of $1/T(m_q,U_q)$ and, by definition of $q$, and from Part~\ref{lempart:TmU_order} of Lemma~\ref{lem:lbs-for-bsprob}, we conclude that $u_b > T(m_q,U_q)$. This bike switch must therefore be possible and we can conclude that requirement $3$ of Definition~\ref{def:feasible} is also satisfied -- provided that a valid partition vector exists.
        
        Similarly to what we did in the proof of Theorem~\ref{thm:all_make_it}, we shall finish the proof of this theorem by demonstrating the existence of a valid unexpanded partition $Z$. This is sufficient because $Z$ can then be expanded into a true partition $X_Z$ that makes $(X_Z,M)$ feasible and every agent completes at the same time (guaranteed by validity of Z). Part~\ref{lempart:tmuy_equals} of Lemma~\ref{lem:lbs-for-rbsprob} then guarantees that the completion time of the agents is $\tau(X_Z,M) = T(X,M,Y=[1,\ldots,1,z_1])$ (since the bike $u_b$ is abandoned at location $z_1$). The fact that $T(X,M,Y=[1,\ldots,1,z_1]) = T_1(m,U)$ follows from our requiring agent $m$ to use the bike $u_b$ until location $z_1$ and then finishing on bike $u_1$, i.e. we explicitly constructed the algorithm so that agent $m$ completes in time $z_1 u_b + (1-z_1) u_1$. Thus the completion time of the algorithm will be $\tau(X_Z,M) = T(X,M,Y=[1,\ldots,1,z_1]) = z_1 u_b + (1-z_1) u_1$, which is precisely the definition of $T_1(m,U)$ (in the case that $u_b > T(m,U)$, see Lemma~\ref{lm:T_1(m,U)}). In summary, we will complete the proof by demonstrating the existence of a valid unexpanded partition $Z$ for the unexpanded matrix $M'$ of the schedule. 
        
        By Claim~\ref{claim:scaling} and the fact that rescaling does not alter the validity of a partition, it will suffice to construct a valid partition over any interval, and we will thus assume that $z_1=1$. For $j=2,3,\ldots,b-q+1$ we choose
        \[z_j = \frac{\sum_{p=1}^{j-1} [\widetilde{M}'(m_{q+j-1},p) - \widetilde{M}'(m_{q+j-2},p)]z_p}{\widetilde{M}'(m_{q+j-2},j)-\widetilde{M}'(m_{q+j-1},j)}\]
        We first show that $z_j \geq 0$ for all $j \in [b-q+1]$ to confirm that $Z$ is indeed a partition of over the interval $[0,\sum_{j=1}^{b-q+1}z_j]$. We do so by demonstrating that each coefficient of $z_p$ is non-negative thereby implying the non-negativity of $z_j$ by a simple induction on $j$.

        Referring to Figure~\ref{fig:all_but_one_make_it} it is easy to see that for $j=2,3,\ldots,b-q+1$ the unexpanded schedule matrix $M'$ is given by
        \[\widetilde{M}'(i,j) = \begin{cases}
            T(m_{q+j-2},U'_{q+j-2}),& 1 \leq i \leq m_{q+j-2}\\
            u_{q+j-1},& m_{q+j-1} \leq i \leq m.
        \end{cases}\]
        From this one can observe that agent $m_{q+j-1}$ uses a bike at least as a fast as agent $m_{q+j-2}$ in every interval preceding $z_j$ and thus we indeed have $\widetilde{M}'(m_{q+j-1},p) \geq \widetilde{M}'(m_{q+j-2},p)$ for each $p=1,2,\ldots,j-1$. We can also see that $\widetilde{M}'(m_{q+j-1},j) = u_{q+j-1}$ and $\widetilde{M}'(m_{q+j-2},j) = T(m_{q+j-2},U'_{q+j-2})$, and we have already demonstrated that $u_{q+i} \leq T(m_{q+i}, U'_{q+i})$ for each $i=0,1,\ldots,b-q-1$. Thus, we have $u_{q+j-1} \leq T(m_{q+j-1},U'_{q+j-1})$ and by Lemma~\ref{cor:fast_set} we have $T(m_{q+j-1},U'_{q+j-1}) \leq T(m_{q+j-2},U'_{q+j-2})$.

        We will now show that $Z$ is valid. Similarly to \eqref{eq:tij_general} we have
        \[t'_{i,j} = \sum_{p=1}^{j} \widetilde{M}'(i,p) z_p\]
        Rearranging our definition of $z_j$ we obtain
        \[(\widetilde{M}'(m_{q+j-2},j)-\widetilde{M}'(m_{q+j-1},j))z_j = \sum_{p=1}^{j-1} [\widetilde{M}'(m_{q+j-1},p) - \widetilde{M}'(m_{q+j-2},p)]z_p\]
        and thus
        \[\sum_{p=1}^{j} \widetilde{M}'(m_{q+j-1},p)z_p - \sum_{p=1}^{j}\widetilde{M}'(m_{q+j-2},p)z_p = t'_{m_{q+j-1},j} - t'_{m_{q+j-2},j} = 0.\]
        Thus, the $j^{th}$ completion times of agents $m_{q+j-1}$ and $m_{q+j-2}$ are the same and the validity of $Z$ follows from a simple induction.
    \end{proof}    

\subsection{Optimal Schedule  when \texorpdfstring{$u_b > T(m,U)$}{ub > T(m,U)} and \texorpdfstring{$u_{b-1} > T_1(m,U)$}{ub1>T1(m,U)}}
    In this section we solve the \RBSP problem with abandonment limit $\ell = 1$ efficiently and optimally for the case when the second slowest bike is the bottleneck. The idea is analogous to what we did for the \BSP problem for the case $u_b > T(m,U)$, namely, we reduce our problem to the previously studied case of Subsection~\ref{ssec:rbs_t1mu}. The following observation is essentially all we need:

    \begin{lemma}\label{lem:ub_gt_t1mu}
        Let $(m, U=\{u_1 \le u_2 \ldots \le u_b\})$ be the input to the \RBSP problem. If $u_b > T(m,U)$ and $u_{b-1} > T_1(m,U)$ then there exists $k \in \{1, 2, \ldots, b-2\}$ such that
        $u_{b-k-1} \le T_1(m-k, U_{b-k-1} \cup \{u_b\}) \le u_{b-1}$.
    \end{lemma}
    \begin{proof}
        We prove the claim by induction on $b$. For $b = 2$ we have $T_1(m, \{u_1, u_2\}) > u_1$ by part~\ref{lempart:min_u1_tmuy} of Lemma~\ref{lem:lbs-for-rbsprob} and the premise of the ``if'' statement is false. Thus, the statement itself is true.

        Next consider $b \ge 3$ and assume the claim is true for any input with $b-1$ bikes. We have $u_{b-1} > T_1(m,U) \geq T_1(m-1,U\setminus \{u_{b-1}\}) = T_1(m-1,U_{b-2} \cup \{b\})$ by Lemma~\ref{lm:T1mU_order}. Thus, if $u_{b-2} \le T(m-1, U_{b-2} \cup \{b\})$ we can set $k = 1$. Otherwise, the claim holds by applying the induction assumption to $(m-1, U_{b-2} \cup \{b\})$.
    \end{proof}
    
    Armed with this lemma, we can now prove the main result of this subsection (Theorems~\ref{thm:rbs_ub_gt_tmu_ub1_gt_t1mu} and \ref{thm:rbs_ub_lt_t1mu}).

    \paragraph*{Proof of Theorem~\ref{thm:rbs_ub_gt_tmu_ub1_gt_t1mu}}
    \begin{proof}
        Let $k$ be as in Lemma~\ref{lem:ub_gt_t1mu}. In the schedule $M$ the first $m-k$ bikers execute the schedule given by \textsc{AllButOne}$(m-k, U_{b-k-1}\cup\{u_b\})$ and the remaining $k$ bikers simply ride bikes $u_{b-k}, u_{b-k+1}, \ldots, u_{b-1}$, respectively, for the entirety of interval $[0,1]$. Note that the choice of $k$ guarantees that input $(m-k, U_{b-k-1}\cup\{u_b\})$ satisfies the preconditions of \textsc{AllButOne} and hence the first $m-k$ bikers have completion time $T_1(m-k, U_{b-k-1}\cup\{u_b\}) \le u_{b-1}$ (the inequality is by Lemma~\ref{lem:ub_gt_t1mu}). The remaining $k$ bikers have completion times $u_{b-k} \le u_{b-k+1}\le \ldots \le u_{b-1}$. Clearly, the overall completion time is exactly $u_{b-1}$. The feasibility of schedule and polynomial time computability follows from the guarantees of \textsc{AllButOne}.
    \end{proof}

    \paragraph*{Proof of Theorem~\ref{thm:rbs_ub_lt_t1mu}}
    \begin{proof}
        Let $(X,M)$ be the schedule produced by $\AllButOne$ and let $(X',M')$ be some other feasible schedule for the \RBSP problem. Let $Y' = Y(X',M') = [y'_1,y'_2,\ldots,y'_b]$. If the bike $b$ is not abandoned in $(X',M')$ then $\tau(X',M') \geq u_b$ and we already know that $u_b > T_1(m,U)$. Thus assume that the bike $b$ is abandoned in $(X',M')$, i.e. $y_b' < 1$. By parts~\ref{lempart:lb_rbs_1} and~\ref{lempart:lb_rbs_2} of Lemma~\ref{lem:lbs-for-rbsprob} we have
        \[\tau(X',M') \ge \max \{ T(m, U, Y'), y_i' u_i + (1-y_i') u_1\} \ge \max\{ T(m, U, Y'), y_b' u_b + (1-y_b') u_1\}\]
        and by Claim~\ref{claim:lb_monotonicity} we have $T(m, U, Y') \ge T(m,U,[1,\ldots,1,y'_b])$. Thus,
        \begin{align*}
            \tau(X',M') &\ge \max\{ T(m, U, [1,\ldots,1,y'_b]),\ y_b' u_b + (1-y_b') u_1\}\\
            &\ge \min_y \max\{ T(m, U, [1,\ldots,1,y]),\ y u_b + (1-y) u_1\}\\
            &= T_1(m,U).
        \end{align*}
    \end{proof}

\section{Agents never need to wait}\label{sec:no-wait}
    In this last section we justify our initial assumption that agents should always travel at their maximum possible speed and never wait. We first generalize the notion of a schedule, to allow a non-zero waiting time at the end of each interval. Note that this representation is general enough to cover those schedules for which agents bike/walk slower than maximum possible speed, since it is trivial to convert such a schedule to one where all agents bike/walk at full speed and instead wait enough time to make up the difference.    We then show that such a schedule can be converted to one in which all waiting times are zero.     
    
    To this end, we modify the representation of a schedule introduced in Section~\ref{sec:prelim} to include a matrix $D$ whose entry $D(i,j)$ indicates an amount of time that the agent $i$ should wait at the end of an interval $j$. We call $D$ the {\em waiting matrix}. According to this kind of augmented schedule $\Alg = (X,M,D)$, agent $i$ rides bike $M(i, j)$ across interval $x_j$, and then waits for time $D(i, j)$. 
        
    The completion time of agent $i$ is now computed from the $i^{th}$ row of $M \cdot X+D$, i.e.
    \[t_{i} = \sum_{j=1}^{n} (M(i,j)x_j + D(i,j)) \]
    Similarly, the $j^{th}$ completion times become
    \[t_{i,j} = \sum_{k=1}^{j} (M(i,k)x_k + D(i,k))\]
    The feasibility requirements need only minimal changes to accommodate $D$. Namely, we need to use the new definitions of $t_{i,j}$ in Definition~\ref{def:feasible} and add the extra requirement that $D$ contain only non-negative entries. Finally, we note that the standardization procedure can be trivially modified to operate on schedules which include a waiting matrix (the only difference is that we now have to delete/merge appropriate columns of $D$ when there are zero/redundant columns, and switch appropriate entries of $D$ when there are swap-switches).
    
    We will now demonstrate that it is never necessary for an agent to wait.
    \begin{lemma}\label{lem:waiting}
        Let $(X,M,D)$ represent a feasible schedule which instructs an agent to wait, i.e. there exists a pair of indices $(i_0,j_0)$ such that $D$ contains an entry $D(i_0,j_0)>0$. Then there exists a feasible schedule $(X',M',D')$ with completion time at least as good as $(X,M,D)$, and such that $D'(i,j) = D(i,j)$ when $(i,j) \neq (i_0,j_0)$ and $D'(i_0,j_0)=D(i_0,j_0)-d$ for some positive constant $d$ bounded away from 0.
    \end{lemma}
    \begin{proof}
        Consider any feasible schedule $\Alg = (X,M,D)$ such that $D$ contains an entry $D(i_0,j_0)>0$ and assume without loss of generality that $\Alg$ is expressed in standard form. Now define the $m \times n$ matrix $S$ called the {\em switch matrix}. This matrix is computed (only) from $M$ and specifies the {\em previous} agent during a pickup switch.  
        \begin{equation}\label{eq:switch}
            S(i,j) = \begin{cases}
                i',& M(i,j) = M(i',j-1) \neq 1,\ i\neq i'\\
                0,& \mbox{otherwise}.
            \end{cases}
        \end{equation}
        In other words, if agent $i$ performs a pickup-switch between $x_{j-1}$ and $x_j$, then $S(i,j)$ is defined as the agent $i'$ that rode the bike across interval $x_{j-1}$ and dropped it to be picked up by agent $i$; otherwise  $S(i,j)=0$. 
            
        Let $t_{i,j}$ represent the $j^{th}$ completion time of agent $i$ using $\Alg$. Since $\Alg$ is feasible and in standard form we know that
        \begin{equation} \label{eq:valid}
            t_{i,j-1}-t_{S(i,j),j-1} > 0 \mbox{ whenever } S(i,j) > 0.
        \end{equation}
        Note that the difference is strictly greater than 0 because $\Alg$ is in standard form (and thus does not contain any swap-switches). 

        Now consider the schedule $\Alg' = (X,M,D')$ where $D'(i,j) = D(i,j)$ for $(i,j) \neq (i_0,j_0)$ and $D'(i_0,j_0) = D(i_0,j_0)-d$ for some $d \in (0,D(i_0,j_0)]$ that we will determine shortly. In other words, the only difference between $\Alg$ and $\Alg'$ is in the entry $(i_0,j_0)$ of the waiting matrix. Let $t'_{i,j}$ represent the $j^{th}$ completion time of agent $i$ using $\Alg'$ and let $S'$ represent the switch matrix of $\Alg'$. Since the only difference between $\Alg$ and $\Alg'$ is in the entry $(i_0,j_0)$ of $D$ and $D'$ it is not hard to see that $S' = S$ and
        \begin{equation}\label{eq:tp_ij}
            t'_{i,j} = \begin{cases}
            t_{i,j},& i \neq i_0 \mbox{ or } j < j_0\\
            t_{i,j}-d,& i = i_0 \mbox{ and } j \geq j_0.
            \end{cases}
        \end{equation}
        From this one can immediately conclude that $\Alg'$ will complete in time no worse (and possibly better) than $\Alg$. It remains only to show that there exists a constant $d \in (0,D(i_0,j_0)]$ bounded away from zero such that $\Alg'$ remains feasible, that is, it satisfies:
        \begin{equation*}
            t'_{i,j-1}-t'_{S(i,j),j-1} \geq 0 \mbox{ whenever } S(i,j) > 0.
        \end{equation*}
        (note that we do not require $\Alg'$ to be in standard form). Using \eqref{eq:tp_ij} we can rewrite these conditions as
        \begin{align}
                &t_{i,j-1} - t_{S(i,j),j-1} \geq 0,& j < j_0 \mbox{ or } (i \neq i_0 \mbox{ and }S(i, j) \neq i_0) \label{eq:firstcond}\\
                &t_{i,j-1} - t_{i_0,j-1} \geq -d,& j \geq j_0 \mbox{ and } S(i,j) = i_0 \label{eq:secondcond}\\
                &t_{i,j-1} - t_{S(i,j),j-1} \geq d,& j \geq j_0 \mbox{ and } i=i_0 \label{eq:thirdcond}
        \end{align} 
        
        Since we already know that $t_{i,j-1}-t_{S(i,j),j-1} > 0$ we can immediately conclude that \eqref{eq:firstcond} and \eqref{eq:secondcond} are satisfied. By choosing 
        \[d = \min\{D(i_0,j_0), \min_{\substack{j \geq j_0\\ S(i_0,j)\neq 0}} t_{i_0,j-1}-t_{S(i_0,j),j-1}\}\]
        clearly \eqref{eq:thirdcond} is satisfied, and  the schedule $\Alg'$ will remain feasible. Note that \eqref{eq:valid}, together with the assumption that $D(i_0, j_0)>0$ assure that $d > 0$. This completes the proof.
    \end{proof}

    \begin{theorem}\label{cor:waiting1}
        Let $\Alg$ represent a schedule  which instructs at least one agent to wait or walk at speed less than 1 or bike with speed less than $v_i$ while using bike $i$. Then there exists a schedule $\Alg'$, with completion time at least as good as $\Alg$, and where no agents are required to wait and all agents move at maximum possible speed. 
    \end{theorem}
    \begin{proof}
        Follows from repeated application of Lemma~\ref{lem:waiting}.
    \end{proof}

\section{Conclusion}\label{sec:conclusion}
   In this paper we introduced  a new paradigm of mobile robot optimization problems involving both autonomous mobile agents and passive transportation devices called bikes that can increase the speed of agents.  We studied the \BSP problem, in which $m$ agents and $b \le m$ bikes are required to traverse a unit interval. We gave a polynomial time algorithm to construct an arrival time-optimal schedule. We also gave a polynomial time algorithm for the \RBSP problem, in the case when at most one bike is allowed to be abandoned. 

    There are many open questions that remain. First, the development of algorithms for the \RBSP problem when  more than one bike can be abandoned is required. The techniques introduced in this paper can be extended further to cover the case that at most 2 bikes can be abandoned, however, this results in a messy case analysis that does not lend any intuition as to how the problem can be elegantly solved. Additionally, one can study more general versions of the problem where agents/bikes do not all begin at the same location, or even versions where the speed of a bike depends both on the bike and on the ID of the agent that is riding it.

\newpage
\bibliographystyle{plain}
\bibliography{allrefs}

\appendix

\section{Pseudocode of procedure \texorpdfstring{$\textsc{Standardize}$}{Standardize}}\label{sec:appendix_bs}
    \begin{algorithm}[H] \caption{$\textsc{Standardize}(X,M)$} \label{alg:standardize}
        \begin{algorithmic}[1]             
        \Input {$m \times n$ schedule matrix $M$, and $n \times 1$ partition vector $X$}
        \Init
            \State $X' \leftarrow []$; \Comment{partition vector of standardized algorithm.}
            \State $M' \leftarrow []$; \Comment{schedule matrix of standardized algorithm.}
            \State $n' \leftarrow 0$; \Comment{size of standardized algorithm}                
            \State $T \leftarrow []$; \Comment{to keep track of the $(j-1)^\text{st}$ completion times of agents.}
        \Begin
        \For {$(j=1$ to $n)$}
            \If {$(X(j)=0)$} { \textbf{continue}; } \Comment{A zero column}  \EndIf 
            \If {$(n'= 0)$} \Comment{First non-zero column, start checking for swap-switches on next iteration}
                \State $X' \leftarrow X(j)$; 
                \State $M' \leftarrow M(1:m,j)$;
                \State $n' \leftarrow 1$;                          
                \State $T \leftarrow T + M(1:m,j) \cdot X(j)$; \Comment{Update $(j-1)^\text{st}$ completion times}
            \Else \Comment{Check for swap-switches}
                \State $i \leftarrow 1$;
                \While {$(i \leq m)$}
                    \If {$(M(i,j) \neq 1$ and $M(i,j) \neq M'(i,n'))$}\Comment{A pickup-switch}
                        \For {$(i'=1$ to $m)$}
                            \If {$(M'(i',j-1)=M(i,j))$}
                                \If {$T(i)=T(i')$}  \Comment{A swap-switch}
                                    \State $\textsc{Swap}(M(i,j:n),M(i',j:n))$; \Comment{Swap schedules of agents $i,i'$}
                                \Else 
                                    \State $i \leftarrow i+1$;
                                \EndIf
                                \State \textbf{break};
                            \EndIf
                        \EndFor
                    \Else { $i \leftarrow i+1$} \EndIf
                \EndWhile
                \If { $M'(1:m,n') = M(1:m,j)$} \Comment {A redundant column}
                    \State $X'(n') \leftarrow X'(n')+X(j)$; \Comment{Merge entries corresponding to redundant columns}
                \Else
                    \State $X' \leftarrow [X'; X(j)]$;
                    \State $M' \leftarrow [M', M(1:m,j)]$;
                    \State $n' \leftarrow n'+1$;                                                
                \EndIf            
                \State $T \leftarrow T + M(1:m,j) \cdot X(j)$; \Comment{Update $(j-1)^\text{st}$ completion times}
            \EndIf              
        \EndFor
        \State \Return $(X',M')$;
        \end{algorithmic}
    \end{algorithm}     

    Algorithm~\ref{alg:standardize} builds up the standardized couple $(X',M')$ column-by-column by looping over the columns of $M$, skipping any zero columns. After the first non-zero column is added to $(X',M')$ we begin checking for swap-switches (lines 13--23), removing any that we find (line 19). Once we have removed all swap-switches from the current column we check if the column is redundant. If not we add it to the couple $(X',M')$ (otherwise we merge the intervals corresponding to the redundant columns, line 25). 
        
    The computational complexity of this implementation is $O(m^2 n)$ since, in each of the $n$ columns, there can be at most $m$ swap-switches and for each of these switches the inner loop on line 16 will run once (at least one swap-switch will be removed on each iteration of this inner loop). The space complexity is clearly $O(mn)$.

\end{document}